%% file: main.tex
\documentclass[sigconf]{acmart}
\usepackage{tikz} 
\usepackage{stfloats}
\usepackage{enumitem}
\usepackage{subfigure}
\usepackage[linesnumbered,ruled,vlined]{algorithm2e}

\AtBeginDocument{%
  }

\setcopyright{acmlicensed}
\copyrightyear{2018}
\acmYear{2018}
\acmDOI{XXXXXXX.XXXXXXX}

\acmConference[Conference acronym 'XX]{Make sure to enter the correct
  conference title from your rights confirmation emai}{June 03--05,
  2018}{Woodstock, NY}

\acmISBN{978-1-4503-XXXX-X/18/06}

\begin{document}
\title{Efficient Multiple Temporal Network Kernel Density Estimation}

\renewcommand{\thefootnote}{\fnsymbol{footnote}}
\author{Yu Shao}
\affiliation{
	\institution{\small East China Normal University}
	\city{Shanghai}
	\country{China}
}
\email{yushao@stu.ecnu.edu.cn}

\author{Peng Cheng}
\affiliation{
	\institution{\small East China Normal University}
	\city{Shanghai}
	\country{China}
}
\email{pcheng@sei.ecnu.edu.cn}

\author{Xiang Lian}
\affiliation{
	\institution{\small Kent State University}
	\city{Kent, Ohio}
	\country{USA}
}
\email{xlian@kent.edu}

\author{Lei Chen}
\affiliation{
	\institution{\small HKUST-GZ and HKUST}
	\city{Guangzhou and Hong Kong SAR}
	\country{China}
}
\email{leichen@cse.ust.hk}

\author{Wangze Ni}
\affiliation{
	\institution{\small HKUST}
	\city{Hong Kong SAR}
	\country{China}
}
\email{wniab@cse.ust.hk}

\author{Xuemin Lin}
\affiliation{
	\institution{\small Shanghai Jiaotong University}
	\city{Shanghai}
	\country{China}
}
\email{xuemin.lin@gmail.com}

\author{Chen Zhang}
\affiliation{
	\institution{\small Hong Kong Polytechnic University}
	\city{Hong Kong SAR}
	\country{China}
}
\email{jason-c.zhang@polyu.edu.hk}

\author{Liping Wang}
\affiliation{
	\institution{\small East China Normal University}
	\city{Shanghai}
	\country{China}
}
\email{lipingwang@sei.ecnu.edu.cn}

\newcommand{\revision}[1]{\color{blue}{#1} \color{black}}

\input{abstract.tex}



\maketitle

\input{introduction.tex}

\input{relatedWork.tex}

\input{problemDefinition.tex}

\input{framework.tex}
\input{proposedSolution1.tex}

\input{proposedSolution2.tex}
\input{proposedSolution3.tex}
\input{proposedSolution4.tex}
\input{experimentalStudy.tex}

\input{conclusion.tex}

\bibliographystyle{ACM-Reference-Format}
\bibliography{main}

\end{document}

%% file: abstract.tex
\begin{abstract}
   Kernel density estimation (KDE) has become a popular method for visual analysis in various fields, such as financial risk forecasting, crime clustering, and traffic monitoring. KDE can identify high-density areas from discrete datasets. However, most existing works only consider planar distance and spatial data. In this paper, we introduce a new model, called TN-KDE, that applies KDE-based techniques to road networks with temporal data. Specifically, we introduce a novel solution, Range Forest Solution (RFS), which can efficiently compute KDE values on spatiotemporal road networks. To support the insertion operation, we present a dynamic version, called Dynamic Range Forest Solution (DRFS). We also propose an optimization called Lixel Sharing (LS) to share similar KDE values between two adjacent lixels. Furthermore, our solutions support many non-polynomial kernel functions and still report exact values. Experimental results show that our solutions achieve up to 6 times faster than the state-of-the-art method.
\end{abstract}

%% file: introduction.tex
\section{INTRODUCTION}

Kernel Density Estimation (KDE) is a popular non-parametric way to smooth data, detect hot spots, and analyze event distribution~\cite{silverman_density_2018, gramacki_nonparametric_2018}, which is widely applied in financial risk forecasting~\cite{diebold_multivariate_1999, diebold_evaluating_1998, harvey_kernel_2012}, crime clustering~\cite{brunsdon_visualising_2007, nakaya_visualising_2010, hart_kernel_2014}, and traffic accident avoiding~\cite{black_highway_1991, xie_kernel_2008, plug_spatial_2011}. It often uses a finite and discrete dataset to generate a smooth distribution.
 The density value of a position $q$ can be estimated from the dataset $O$ by:
\begin{equation}
	F(q) = w \cdot \sum_{o_i \in O} {K}\Big(\frac{dist(q, o_i)}{b}\Big),
\end{equation}
where $w$ is a scaling factor defaulted as $1$ and ${K}(\cdot)$ is the kernel function, which can be polynomial kernel function (e.g., Triangular~\cite{fleuret_scale-invariance_2003, gong_estimating_2014}  and Epanechnikov~\cite{samiuddin_nonparametric_1990, bil_identification_2013}), or transcendental kernel function (e.g., Gaussian~\cite{scholkopf_comparing_1997, kristan_multivariate_2011}  and Cosine~\cite{de_felice_short-term_2015}). Table \ref{tab1} summarizes some popular kernel functions and their expressions, where $dist(q, o_i)$ indicates the Euclidean distance.

\begin{table}[!h]
	\centering
	\caption{Kernel Functions}
	\vspace{-0.5em}
	\label{tab1}
	\begin{tabular}{c|c}
		\hline
		Kernel       & Function                                             \\ \hline
		Triangular \cite{fleuret_scale-invariance_2003, gong_estimating_2014}   & $1-\frac{1}{b}dist(q,o_i)$         \\ 
		Epanechnikov \cite{samiuddin_nonparametric_1990, bil_identification_2013} & $1-\frac{1}{b^2}dist(q,o_i)^2$     \\ 
		Gaussian \cite{scholkopf_comparing_1997, kristan_multivariate_2011}     & $exp(-\frac{1}{b^2}dist(q,o_i)^2)$ \\ 
		Cosine \cite{de_felice_short-term_2015}       & $\cos(\frac{1}{b}dist(q,o_i))$     \\ \hline
	\end{tabular}
\end{table}

The density value conveys the combined impact of events within the bandwidth. The weight assigned to closer events is higher, which is determined by the kernel function. A basic distance measure is the Euclidean distance, implying that the impact of events spreads linearly.

\begin{figure}[t!]
	\centering
	\subfigure[Heatmap in 3 A.M. -- 6 A.M.]{\scalebox{0.35}[0.35]{\includegraphics{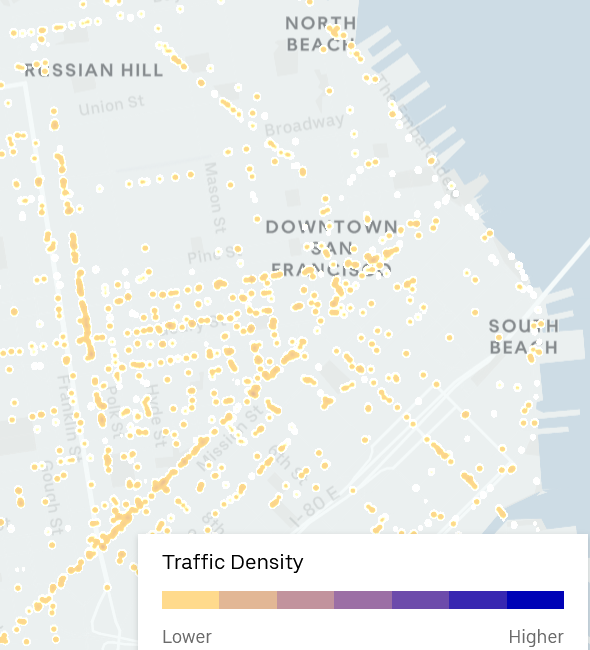}}
		\label{subfig:morning_heatmap}}
	\subfigure[Heatmap in 3 P.M. -- 6 P.M.]{\scalebox{0.35}[0.35]{\includegraphics{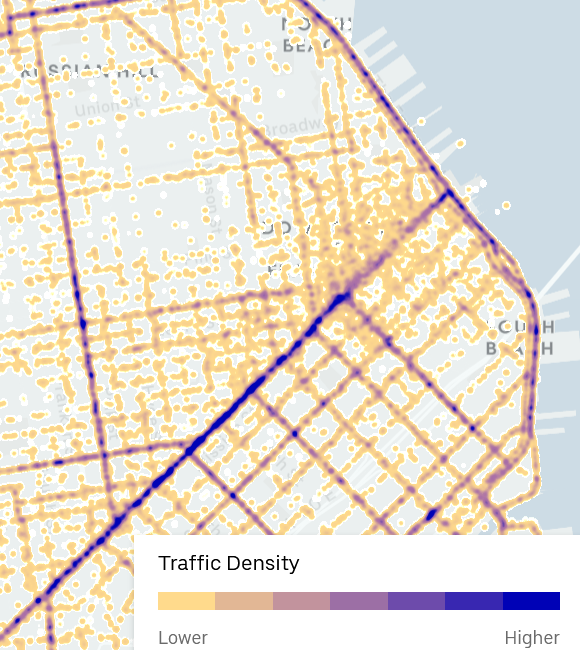}}
		\label{subfig:afternoon_heatmap}}
	\caption{Mobility Heatmaps in San Francisco Area}
	\label{fig:temporal_heatmap}
\end{figure}

\noindent\textbf{Challenges.} Since the KDE-based algorithm is easy to understand, implement and visualize, it has been widely supported in many geographic software (e.g., ArcGIS \cite{noauthor_arcgis_nodate}, QGIS \cite{noauthor_qgis_nodate}, and KDV-Explorer\cite{chan_kdv-explorer_2021}). However, the original KDE has limitations in the following three applications:

\begin{itemize}[leftmargin=*]

	\item Temporal clustering: KDE only considers spatial distance to measure event density. However, many events are time-sensitive \cite{black_highway_1991}. Users may filter events in a specific time period to reveal temporal relations \cite{brunsdon_visualising_2007}. For instance, Uber Movement displays mobility heatmaps for selected cities during specific time periods. Figure \ref{fig:temporal_heatmap} shows two traffic mobility heatmaps of the San Francisco area in Q1 2020, from 3 A.M. to 6 A.M. and from 3 P.M. to 6 P.M., respectively.

	\item Network application: Xie et al. \cite{xie_kernel_2008} found that planar KDE can overestimate density values in networks. Figure \ref{fig:kde_example} compares planar and network KDE. An event point that is $50m$ away in Euclidean distance is actually $70m$ in shortest path distance due to the difference between straight and network distances.
	
	\item Multiple and online queries: Users often call multiple online queries with different parameters to generate a set of multiple KDEs, then explore them to select the best smooth KDE. The parameter selectors usually update parameters based on the errors from the last estimations, which requires real-time results. However, existing solutions are inefficient in generating each KDE on networks from raw data independently \cite{gong_estimating_2014, gan_scalable_2017, plug_spatial_2011, brunsdon_visualising_2007, chan_safe_2021, cristianini_dynamically_1998, gramacki_nonparametric_2018}.
\end{itemize}

\begin{figure}[t!]
	\centering
	\subfigure[KDE on Euclidian Distance]{\scalebox{0.6}[0.6]{\includegraphics{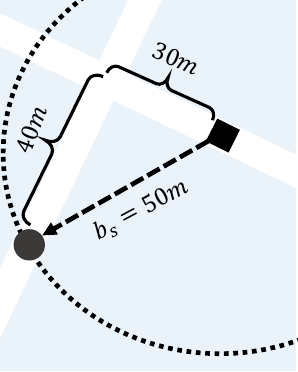}}
		\label{subfig:loss_13}}
	\hspace{1em}
	\subfigure[KDE on Network Distance]{\scalebox{0.6}[0.6]{\includegraphics{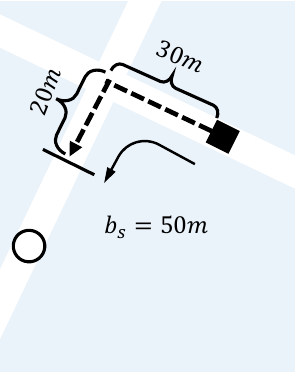}}
		\label{subfig:loss_12}}
	\vspace{-1em}
	\caption{An example of over-estimation bandwidth.}
	\vspace{-2em}
	\label{fig:kde_example}
\end{figure}

Existing studies have identified ways to improve Kernel Density Estimations (KDEs). To support temporal estimation, Charlotte et al.~\cite{plug_spatial_2011} created KDEs by hour, day, or month to study the relationship between events over time. Another common method is to use a space-time cube in 3D space \cite{nakaya_visualising_2010, black_highway_1991} and estimate density values with the space-time kernel function \cite{brunsdon_visualising_2007, romano_visualizing_2017, chan_sws_2021}.
To fit the network, instead of Euclidean distance, we can use network distance \cite{borruso_network_2005} and compute the kernel density value per linear unit (called lixel) along edges instead of per pixel on the Euclidean space \cite{xie_kernel_2008}. To speed up repeated computation, Chan et al. \cite{chan_fast_2021} build a list index on each edge. 
To handle multiple and online queries, previous studies utilize partition methods (e.g., K-d tree~\cite{chan_efficient_2020, chan_quad_2020, chan_karl_2019}, Gridding~\cite{hart_kernel_2014, black_highway_1991}, Fast Fourier Transform (FFT)~\cite{silverman_algorithm_1982, gramacki_nonparametric_2018}, Clustering~\cite{auber_interactive_2005, abello_ask-graphview_2006, hinneburg_denclue_2007}, and Binning~\cite{liu_immens_2013, gramacki_nonparametric_2018, li_interactive_2014}) to first divide the whole datasets into groups and then aggregate similar events to compute their density value contributions \cite{liu_immens_2013}.

Although these works have some optimizations, no single method can efficiently resolve all three major limitations. Thus, we propose Temporal Network Kernel Density Estimation (TN-KDE) for networks with spatiotemporal datasets.

\noindent\textbf{Contributions.} Specifically, in this paper, we introduce Range Forest Solution (RFS), an efficient method for computing KDE values within a specified query time window. RFS employs a memory-shared and tree-based structure to maintain events in both spatial and temporal dimensions. This innovative approach enables the handling of multiple temporal queries while significantly reducing memory consumption. 

Additionally, we propose Dynamic Range Forest Solution (DRFS), which extends RFS by incorporating a dynamic structure to support insertion operations. DRFS also offers users the ability to fine-tune the quantization of the structure to various precision levels. This flexibility allows for customization of index size and processing time, enabling users to optimize memory and performance based on their specific requirements.

Furthermore, leveraging the observation that two adjacent lixels frequently share a similar KDE values, we develop an optimization technique called Lixel Sharing (LS), which enables the computation of KDE values for shared lixels simultaneously.

Finally, our framework supports more complicated kernel functions, including the Exponential kernel function and the Cosine kernel function, to compute exact KDE values.

In this paper, we make the following contributions:

\begin{itemize}[leftmargin=*]
	\item  We formally define the temporal network kernel density estimation (TN-KDE) problem and introduce a basic framework in Section~\ref{sec3:preliminaries}.
	
	\item We propose two solutions, namely Range Forest Solution (RFS) in Section \ref{sec5:solution} to efficiently maintain events in both spatial and temporal dimensions, as well as Dynamic Range Forest Solution (DRFS) in Section \ref{sec7:dynamic} to support insertion operation.
	
	\item We develop the Lixel Sharing (LS) optimization technique in Section \ref{sec6:LixelsAggregaion} to reduce redundancy and enhance  efficiency. We analyze many non-polynomial kernel functions that can be applied in our framework in Section \ref{sec7:kernel}.
	
	\item We test on real-world datasets with various scales and categories to show our proposed methods' efficiency and effectiveness in Section \ref{sec8:exp}.
\end{itemize}

%% file: relatedWork.tex
\section{Related Work}

\noindent\textbf{Partition-based methods}.
KDE-based algorithms can be time-consuming \cite{gramacki_nonparametric_2018}. To avoid this, a popular method is aggregation, where homologous events are aggregated and their density value contribution is computed in constant time~\cite{liu_immens_2013}. Efficient partition methods, such as:

\begin{itemize}[leftmargin=*]
	\item \textbf{K-d tree}~\cite{chan_efficient_2020, chan_quad_2020, chan_karl_2019}. The K-d tree is one of the most popular structures for organizing points in multi-dimensional space. It partitions the hyperspace in the middle, repeating in each dimension in turn. Each non-leaf node is an aggregation of a hypercube, and each leaf node represents an event point. However, the K-d tree is effective only in low dimensions (around 10). For higher dimensions, the ball tree is more appropriate~\cite{gray_nonparametric_2003}.
	
	\item \textbf{Gridding}~\cite{hart_kernel_2014, black_highway_1991}. Gridding involves dividing the event space into equally spaced grids. Each event will be assigned to the located grid. Density values are computed at the center of grids, not at events. The challenge is to choose the appropriate gridding size. If the size is large, the error between the event and the grid center will also be large. However, if the size is small, many grids may be empty, leading to an even larger time complexity.
	
	\item \textbf{Fast Fourier Transform (FFT)}~\cite{silverman_algorithm_1982, gramacki_nonparametric_2018}. After gridding, FFT can be used to speed up the computation of kernel density values. This involves rewriting the kernel equation as a kernel matrix multiplied by an event matrix. FFT is particularly useful for large kernels and bandwidths. However, it may be computationally expensive for sparse matrices~\cite{fan_fast_1994}.
	
	\item \textbf{Clustering}~\cite{auber_interactive_2005, abello_ask-graphview_2006, hinneburg_denclue_2007}. Clustering groups nearby events by replacing them with a new point, making it effective for sparse events. The algorithm generates a cluster tree during repetition, which can be used for interactive visualization.
	
	\item \textbf{Binning}~\cite{liu_immens_2013, gramacki_nonparametric_2018, li_interactive_2014}. Binning is a method of creating unequal grids by dividing a continuous range into adjacent intervals. The intervals do not have to be equally spaced. Binning assigns a weight to represent events in a bin~\cite{fan_fast_1994} instead of a constant point to replace them. This weighted aggregation is more accurate and can provide an exact solution.
\end{itemize}

\noindent\textbf{Temporal-based methods.}	
Temporal data analysis has been widely studied. To update temporal data, there are three modes:
	
	\begin{itemize}[leftmargin=*]
		\item Static data~\cite{patuelli2007network, mouratidis2006continuous}: Static data is permanent and unchanging. Temporal attribution can be seen as a new dimension and an index can be built on it. Most methods can only handle static data, such as PAS and RFS in this paper.
		\item Persistent data~\cite{sasikala2014uncertain, cheng2012spatio}: Persistent data can be updated anytime and anywhere. Basic methods can only support update and query operations, while more complex indexes are required for others.
		\item Streaming data~\cite{koudas2004approximate, figueiras2018real, li2021trace}: Streaming data is a special type of persistent data where new updates only appear in the latest. It has fewer requirements than persistent data but is still practical if new data consistently arrives over time.
	\end{itemize}
	
\noindent\textbf{Kernel function computation.}
	Weighted aggregation is an important operation in KDE algorithms. Previous studies could only aggregate data with polynomial kernel functions, while other non-polynomial kernel functions like Cosine and Exponential can only be estimated~\cite{chan_karl_2019, chan_quad_2020}. This is done by replacing values with minimum and maximum KDE values or using a polynomial bound function as a substitute. However, there is currently no exact method to calculate the KDE value.

%% file: problemDefinition.tex
\section{Preliminaries}
\label{sec3:preliminaries}

In this section, we first give the formal definition of the TN-KDE problem along with some state-of-the-art solutions. Some key notations are listed in Table~\ref{tab:symbols}.

\subsection{Problem Definition}

\begin{figure}[t!]\centering
	\scalebox{0.6}[0.6]{\includegraphics{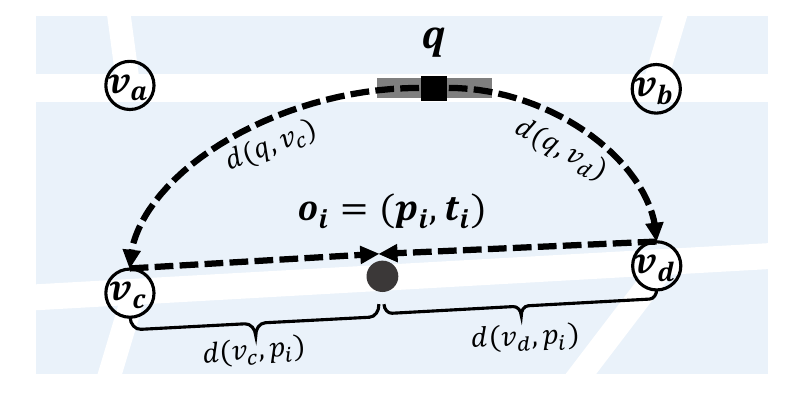}}
	\vspace{-1em}
	\caption{An example of a query $q$ (denoted by black square) and an event $o_i$ (denoted by black dot).}
	\label{fig:sketch}
\end{figure}

\begin{definition}[Road Network]
A road network is a graph $G = (V, E)$, where $V$ is the vertex set and $E \subseteq V \times V$ is the edge set. We use $d(v_a, v_b)$  to denote the shortest distance from $v_a$ to $v_b$ on $G$. 
\end{definition}

To generate a KDE on the given road network, we follow existing studies \cite{xie_kernel_2008, chan_fast_2021} that each road is divided into a set of same-length segments, called lixels.

\begin{definition}[Lixel]
	Given a road network $G = (V, E)$, each edge $(v_a,v_b) \in E$ is divided into small segments using the same spatial length $g$.
\end{definition}

Each lixel is a basic query unit denoted by $q$ that requires a KDE value, while its position is the center point.

The total number of lixels is
$L = \sum_{(v_a,v_b) \in E} \left\lceil \frac{d(v_a,v_b)}{g} \right\rceil$. In Figure \ref{fig:sketch}, the gray segment is a lixel on edge $(v_a,v_b)$ and the center is marked by a black square.

\begin{definition}[Event]
	An event $o_i=(p_i, t_i) \in O$ happens at position $p_i$ at time $t_i$. The number of events on the edge $e$ is $n_e$.
\end{definition}

TN-KDE aims to generate a Kernel Density Estimation (KDE) on a road network within a given spatial and temporal range:

\begin{definition}\textit{(Temporal Network Kernel Density Estimation, TN-KDE)}
	Given a road network $G=(V, E)$, an events set $O$, the spatial bandwidth $b_s$ and the temporal bandwidth $b_t$, the problem of temporal network kernel density estimation (TN-KDE) is to compute KDE values for all lixels $F(q)$ with the time range $[t - b_t, t + b_t]$:
\begin{equation}
\label{eq-def}
\begin{aligned}
	F(q) &= \sum_{o_i=(p_i, t_i) \in O} f(q, o_i), \\
	f(q, o_i) &= 
	K_s\left(\frac{d(q, p_i)}{b_s}\right)
	K_t\left(\frac{\vert t - t_i \vert}{b_t}\right),
\end{aligned}
\end{equation}
where the functions $K_s(\cdot)$ and $K_t(\cdot)$ are kernel functions (which can be chosen from Table \ref{tab1} in arbitrary).
The spatial distance is the shortest path distance $d(q, p_i)$ and the temporal distance is the time difference $\vert t - t_i \vert$. Note the domain of kernel functions is $[0, 1]$, thus we do not consider events out of bandwidth.
When all lixels queries of a TN-KDE are calculated, a heatmap (similar to Figure \ref{fig:temporal_heatmap}) of the road network is achieved.
\end{definition}


\begin{table}[t!]
	\centering
	\caption{Symbols and their descriptions.}
	\vspace{-1em}
	\begin{tabular}{l|l}\label{tab:symbols}
	\textbf{Symbol} & {\textbf{Description}}\\ \hline\hline
	$G=(V,E)$      & a graph $G$ with vertices $V$ and edges $E$ \\
	$q$	     	   & the query lixel \\
	$t$		       & the query time \\
	$o_i = (p_i, t_i)$ & an event $o_i$ at position $p_i$ at time $t_i$ \\
	$d(u, v)$      & the shortest path distance from $u$ to $v$ \\
	$b_s, b_t$     & the spatial and temporal bandwidth \\
	$L, N$		   & the number of lixels and events \\
	$n_e$          & the number of events on edge $e$ \\
	\hline
	\end{tabular}
\end{table}

\subsection{The State-of-the-Art Solutions}

To the best of our knowledge, Aggregate Distance Augmentation (ADA)~\cite{chan_fast_2021} is the state-of-the-art solution for massive events. The original framework of ADA only supports spatial distance so we set the Epanechnikov kernel function as $K_s$ and ignore the temporal kernel function $K_t$, i.e.,
\begin{equation*}
	f(q, o_i) = 1 - \frac{1}{b_s^2} dist(q, p_i)^2
\end{equation*}

The core step of ADA is to rewrite Equation~\ref{eq-def} as:
\begin{equation}
\label{eq-transform}
\begin{aligned}
	F(q) &= \sum_{e \in E} F_e(q) \\
	F_e(q) &= \sum_{o_i \in O_e} 1 - \frac{1}{b_s^2} d(q, p_i)^2
\end{aligned}
\end{equation}
where $O_e$ contains all events on the edge $e$. This transformation decomposes the total KDE value $F(q)$ onto each edge as $F_e(q)$.

Figure~\ref{fig-ADA} shows the illustration of the ADA method on a specific edge. Suppose the lixel $q$ locates on the edge $(v_a, v_b)$ and $e=(v_c, v_d)$, the shortest path from $q$ to $p_i$ must go through either $v_c$ or $v_d$. Take $v_c$ as an example, ADA has:
\begin{equation*}
\begin{aligned}
	F_\Gamma(q) &= \sum_{o_i \in O_\Gamma} \left(1 - \frac{[d(q, v_c) + d(v_c, p_i)]^2}{b_s^2}\right) \\
	&= \frac{1}{b_s^2} \cdot
	\begin{bmatrix}
		-1 \\
		-2 \cdot d(q, v_c) \\
		b_s^2 - d(q, v_c)^2
	\end{bmatrix}^\top
	\begin{bmatrix}
		\sum_{o_i \in O_\Gamma} d(v_c, p_i)^2 \\
		\sum_{o_i \in O_\Gamma} d(v_c, p_i) \\
		\vert O_\Gamma \vert \\
	\end{bmatrix},
\end{aligned}
\end{equation*}
\begin{figure}[t]\centering
	\scalebox{0.6}[0.6]{\includegraphics{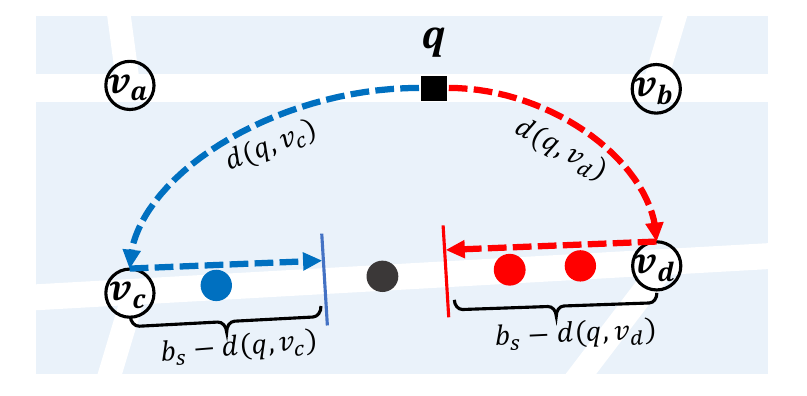}}
	\vspace{-1em}
	\caption{The illustration of the ADA method.}
	\vspace{-1em}
	\label{fig-ADA}
\end{figure}
where $O_\Gamma$ is the event aggregation including all events closer to $v_c$. The first vector along with $\frac{1}{b_s^2}$ (denoted by $\mathbf{Q}$) is constant for the lixel $q$, while the second vector (denoted by $\mathbf{A}$) is an aggregation of some continuous events.

Observing that all items in $\mathbf{A}$ follow the associative law, we can precompute the aggregation value linearly, denoted by $\mathbf{A}_1, \mathbf{A}_2 ... \mathbf{A}_{n_e}$, where $n_e$ is the number of events on $e$. Therefore, for each lixel $q$, ADA can use $b_s - d(q, v_c)$ to find the boundary by the binary search, i.e., the largest index $i$ that $d(v_c, p_i) \le b_s - d(q, v_c)$, and compute the density value by $F_\Gamma(q) = \mathbf{Q} \cdot \mathbf{A_i}$.

The other side of $v_d$ is similar except for that the index order is reversed. Figure~\ref{fig-ADA} displays two aggregations from $v_c$ (marked by blue) and $v_d$ (marked by red), respectively. Since two aggregations might be overlapped for a larger bandwidth, another boundary $d(v_c, p_i) \le \frac{- d(q, v_c) + d(q, v_d) + d(v_c, v_d)}{2}$ (i.e., breakpoint or middle point) is required, which means events should choose a shorter path if both sides are available.

The only remaining challenge lies in the time-consuming process of computing the shortest path distance $d(q, v_c)$ for different lixels $q$. Instead, most recent works adopt the Shortest Path Sharing~(SPS) method~\cite{rakshit2019efficient} to share the shortest path distance. Specifically, the shortest path from $q$ to $v_c$ is either $q {\rightarrow} v_a {\rightarrow} v_c$ or $q {\rightarrow} v_b {\rightarrow} v_c$. Once we precompute $d(v_a, v_c)$ and $d(v_b, v_c)$, other lixels can reuse these partial shortest paths.

\begin{lemma}
	\label{lemma:ADA}
	The time complexity of the ADA method is $O(\vert E \vert \cdot T_{sp} + L \cdot \vert E \vert \cdot \log(\frac{N}{\vert E \vert}))$~\cite{chan_fast_2021}.
\end{lemma}

\begin{proof}
	The shortest paths are shared for all lixels on the same edge, costing $\vert E \vert T_{sp}$. The aggregated vector $\mathbf{A}$ requires binary search, costing $L \cdot \sum\log(n_e) \le L \cdot \vert E \vert \cdot \log(\frac{N}{\vert E \vert})$ using AM-GM inequality.
\end{proof}

%% file: framework.tex
\subsection{Framework}
\label{sec3.3:framework}

The TN-KDE problem is more complex since the temporal distance $\vert t - t_i \vert$ is hard to aggregate. Instead, we continue to divide events into doubled aggregations in which $t_i < t$ and $t_i \ge t$, respectively. Considering the aggregation in which events are closer to $v_c$ and the time $t_i < t$, suppose the Triangular kernel function is applied, the aggregated density value of $O_\Gamma$ can be updated as:
\begin{align}
	\label{eq:aggregate}
	&F_\Gamma(q) \notag \\
	=& \sum_{o_i \in O_\Gamma}
	\left(1 - \frac{d(q,p_i)}{b_s}\right)
	\left(1 - \frac{\vert t - t_i \vert}{b_t}\right) \notag \\
	=& \sum_{o_i \in O_\Gamma}
	\left(1-\frac{d(q, v_c) + d(v_c, p_i)}{b_s}\right)
	\left(1-\frac{t - t_i}{b_t}\right) \\
	=& \frac{1}{b_sb_t} \cdot 
	\begin{bmatrix}
		-1 \\
		b_s-d(q, v_c) \\
		-(b_t-t) \\
		(b_s-d(q, v_c))(b_t-t) \\
	\end{bmatrix}^\top \cdot
	\begin{bmatrix}
		\sum_{o_i \in O_\Gamma} d(v_c,p_i)t_i \\
		\sum_{o_i \in O_\Gamma} d(v_c,p_i) \\
		\sum_{o_i \in O_\Gamma} t_i \\
		\vert O_\Gamma \vert 
	\end{bmatrix} \notag
\end{align}

The first $\mathbf{Q}$ vector is still constant for the lixel $q$, while the second $\mathbf{A}$ vector is the aggregation of continuous events with different items. Algorithm~\ref{algo:framework} is the framework of TN-KDE. For each lixel $q$ and edge $e$, TN-KDE retrieves the query vector $\mathbf{Q}$ and the aggregated vector $\mathbf{A}$ to compute the aggregated density value $F_\Gamma(q)$. These aggregated density values are accumulated as the result of $F(q)$.
\begin{algorithm}[h]
	\caption{TN-KDE Framework}
	\label{algo:framework}
	\DontPrintSemicolon
	\SetKwComment{comment}{$\triangleright$ }{}
	
	\For{each lixel $q$ }{
		\For{each edge $e=(v_c, v_d) \in E$}{
			Get shared shortest path $d(q, v_c)$ and $d(q, v_d)$ \\
			\For{each aggregation $\Gamma$ on edge $e$}{
				$\mathbf{Q} \leftarrow$ the query vector 

				$\mathbf{A} \leftarrow$ the aggregated vector

				$F_{\Gamma}(q) \leftarrow \mathbf{Q} \cdot \mathbf{A}$ \\
				$F_e(q) \leftarrow F_e(q) + F_{\Gamma}(q)$
			}
		}
		$F(q) \leftarrow F(q) + F_e(q)$
	}
\end{algorithm}

This framework shares similarities to the ADA method, but the primary challenge lies in efficiently retrieving aggregations~(Line 4 in Algorithm~\ref{algo:framework}). While ADA relies on a simple binary search to obtain two aggregations, TN-KDE necessitates a more complex approach that accounts for both spatial and temporal dimensions simultaneously. In the subsequent sections, we will explore efficient indexes on each edge to maintain these aggregations without incurring additional time costs.

%% file: proposedSolution1.tex
\section{Efficient Aggregation Method}
\label{sec5:solution}

In this section, we propose an efficient aggregation method called the Range Forest Solution to maintain events on each edge, allowing for rapid retrieval of the aggregation with various query times and distances.

\subsection{Range Forest Solution}
\label{subsec:RFS}
	
	The range tree~\cite{de2000computational} is an efficient tree-based index that maintains events hierarchically. Each tree node maintains an aggregation of some events, allowing for quick retrieval.
	Initially, all events are sorted according to their relative position~(e.g., $d(v_c, p_i)$) and the root node encompasses all events. Then each tree node will be recursively divided into two subsequent aggregations, each of which maintains half of the events. Specifically, a tree node $u$ with events $\{o_l,...,o_r\}$ has two child nodes $lc(u)$ and $rc(u)$ with events $\{o_l,...,o_{\lfloor (l + r) / 2 \rfloor}\}$ and $\{o_{\lfloor (l + r) / 2 \rfloor + 1},...,o_r\}$, respectively. Each tree node also has a range to represent the aggregation, e.g., $R(u) = [d(v_c, o_l), d(v_c, o_r)]$.
	
	$T_4$ in Figure~\ref{fig:RFS1} is an example of the range tree with four events $o_1, o_2, o_3, o_4$ on the edge $(v_c, v_d)$ from left to right. The root node contains all four events, while each child node inherits half of the events.

	The range tree only contains spatial information. To support temporal queries, we progressively add events into the range tree according to their respective time $t_i$. All intermediate states are persistently stored to form a series of range trees, called the range forest. Figure~\ref{fig:RFS1} shows sequential states following each insertion, with the time order $o_1, o_3, o_4, o_2$. The updated tree nodes are highlighted in blue, while other unmodified nodes remain unchanged. Notably, the range $R(u)$ is constant during the insertion process though some events may have not been appeared.
	
\begin{figure}[h!]\centering\vspace{-2ex}
\scalebox{0.33}[0.33]{\includegraphics{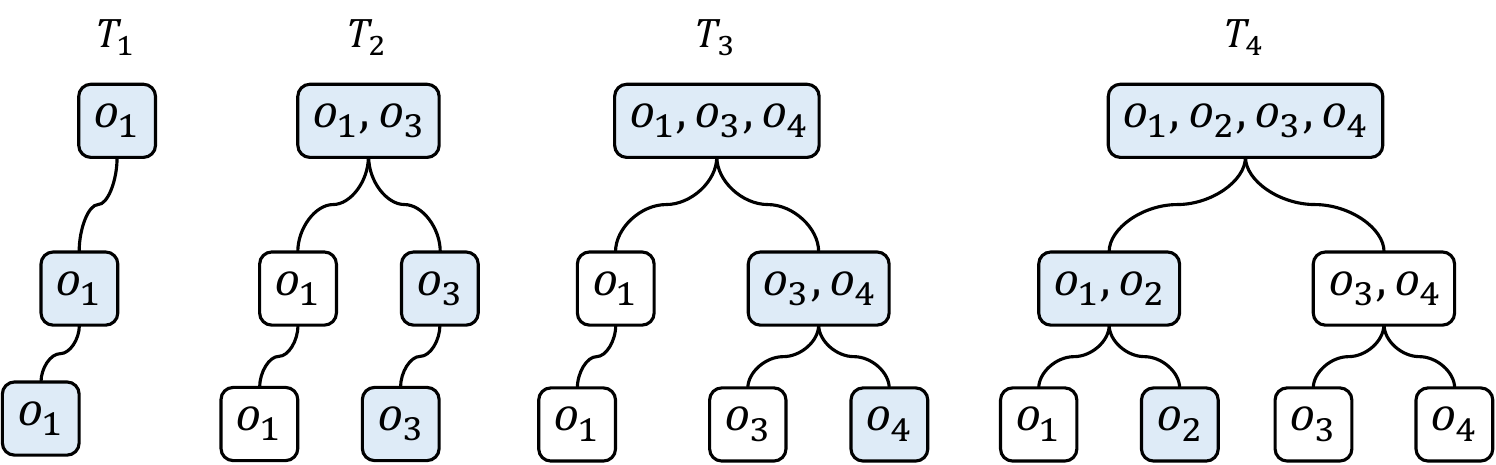}}
\caption{A range forest with four event points. The insertion order is $\{o_1, o_3, o_4, o_2\}$, forming four range trees $T_1,T_2,T_3,T_4$.}
\label{fig:RFS1}
\end{figure}

\subsection{Query on the Range Forest}

	The query process on a range forest consists of two steps, \textit{generating} and \textit{detecting}.

	\noindent\textbf{Generating.} By subtracting two range trees, we can create a new range tree that contains all events within a time period. This approach allows us to generate a range tree consisting of events in the querying time window. Figure~\ref{fig:RFS2} provides an example of a time window containing two events $o_3,o_4$ by subtracting $T_1$ from $T_3$.

\begin{figure}[h!]\centering\vspace{-2ex}
	\scalebox{0.35}[0.35]{\includegraphics{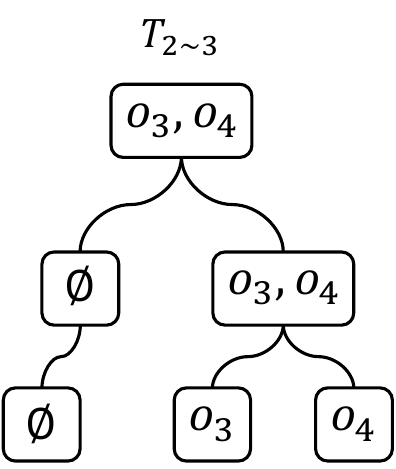}}
	\caption{An example of a query on the range forest, where the time window contains two events $o_3$ and $o_4$.}
	\label{fig:RFS2}
\end{figure}

	\noindent\textbf{Detecting.} After generating a range tree, we then detect this range tree with the aggregation boundary. Considering the $v_c$ side as an example, the boundary has two conditions: (1) the position cannot exceed the spatial bandwidth $b_s$ and (2) the position is closer to $v_c$ than $v_d$, which are:
	\begin{equation}
	\begin{aligned}
		d(v_c, p_i) &\le b_s - d(q, v_c) \\
		d(v_c, p_i) &\le \frac{-d(q, v_c) + d(q, v_d) + d(v_c, v_d)}{2}
	\end{aligned}
	\end{equation}
	
	Let $R {=} [0, \min(b_s - d(q, v_c), \frac{-d(q, v_c) + d(q, v_d) + d(v_c, v_d)}{2})]$ be the spatial range. We aim at retrieving all events where $d(v_c, p_i) \in R$. This problem can be solved by a recursive search shown in Algorithm~\ref{algo:RFQ}. We track the tree nodes $u_l$ and $u_r$ on two trees $T_{l-1}$ and $T_r$ to simulate the subtraction. $R(u_r)$ (as well as $R(u_l)$) represents the current tree node range. If $R(u_r) \cap R = \varnothing$~(Line 4), all events are out of the range so return a zero-vector. If $R(u_r) \subseteq R$~(Line 6), all events are within the range so return the subtraction of the aggregated vector $\mathbf{A}(u_r) - \mathbf{A}(u_l)$. Otherwise, search in two child nodes in detail.

	\begin{algorithm}[t]
		\caption{Range Forest Query}
		\label{algo:RFQ}
		\DontPrintSemicolon
		\KwIn{the temporal query range $[T_l, T_r]$ \\ \hspace{2.75em} the spatial query range $R$}
		\KwOut{the aggregated vector $\mathbf{A}$}

		$u_l \leftarrow root(T_{l-1}), u_r \leftarrow root(T_r)$
		
		\Return{$DualDetect(u_l, u_r)$}
		
		\SetKwFunction{DQ}{DualDetect}
		\SetKwProg{Fn}{Function}{:}{}
		\Fn{\DQ($u_l, u_r$)}{
		
		\uIf{$R(u_r) \cap R = \varnothing$}{
			\Return{$\mathbf{0}$}
		}
		\ElseIf{$R(u_r) \subseteq R$}{
			\Return{$\mathbf{A}(u_r) - \mathbf{A}(u_l)$}
		}
		\Else{
			\Return{
				$DualDetect(lc(u_l), lc(u_r)) + DualDetect(rc(u_l), rc(u_r))$
			}
		}
		}
	\end{algorithm}

\begin{lemma}
	Algorithm~\ref{algo:RFQ} takes $O(\log{n_e})$ time to find the aggregated vector of events in the query range on the edge $e$, where $n_e$ is the number of events on this edge.
\end{lemma}

\begin{proof}
	Computing the temporal range $[T_l, T_r]$ requires a binary search costing $O(\log{n_e})$. As for the \textit{DualDetect} function, each node has three cases: (1) not covered~(Line 4), (2) fully covered~(Line 6), and (3) partially covered~(Line 8). For a partially covered node $u$, if $lc(u)$ is partially covered, $rc(u)$ is not covered. If $rc(u)$ is partially covered, $lc(u)$ is fully covered. Therefore, each layer of the range tree has at most $1$ partially covered node and only $2$ nodes will be accessed. Since the range tree has a balanced structure and the depth is $\log{n_e}$, the query time complexity in Algorithm~\ref{algo:RFQ} is $O(\log{n_e})$.
\end{proof}

\subsection{Construction of the Range Forest}

	If we construct the whole range forest as Figure \ref{fig:RFS1}, the memory consumption is not affordable. An observation reveals that two adjacent range trees exhibit significant structural similarities with a large number of shared nodes (transparent nodes), differing only in $O(\log n_e)$ nodes from the root to the leaf highlighted in blue.
	Consequently, we only need to update these new nodes and link other old nodes to the previous ones, as shown in Figure \ref{fig:RFS3}. This sharing structure optimization only modifies node references without altering the underlying logical relationships.  As a result, the query process in Algorithm \ref{algo:RFQ} remains unchanged.
	
	\begin{figure}[h!]\centering
		\scalebox{0.33}[0.33]{\includegraphics{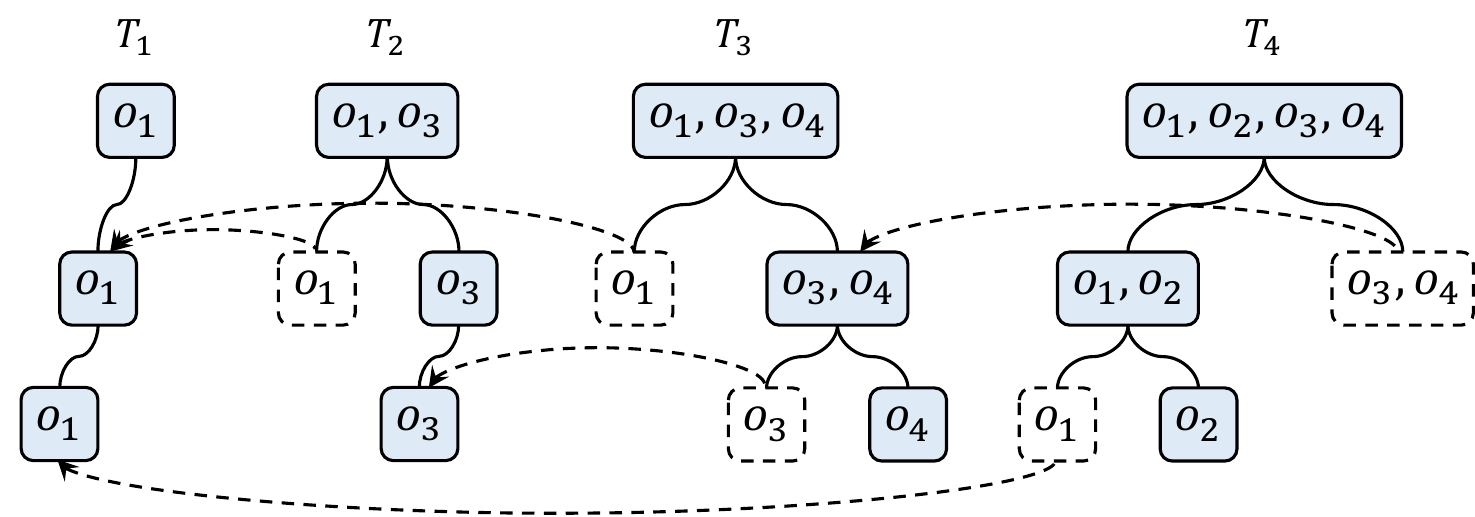}}
		\vspace{-1em}
		\caption{A shared version of the range forest. Only updated nodes (marked by blue) will be accessed and other unchanged nodes are linked to the previous range tree.}
		\label{fig:RFS3}
	\end{figure}

	Algorithm \ref{algo:RFC} describes the process of constructing an optimized range forest. In Lines 2-4, some range trees $T_i$ are constructed by the time order. In each construction, nodes are dynamically created in Line 6, while the aggregated vector is updated in Line 7. If $o_i$ will appear in the left child node $lc(u_r)$, i.e., $d(v_c, p_i) \in R(lc(u_l))$, the right child node will link to the previous one in Line 9. The opposite case follows the same procedure.

\begin{lemma}
	Algorithm \ref{algo:RFC} constructs a range forest with shared tree nodes in $O(n_e\log{n_e})$ time and space.
\end{lemma}

\begin{proof}
	The shared version includes $n_e$ spatial range trees. In the dual construction, the algorithm will always recurse to one child and link the other child. This step will execute $\log{n_e}$ times for each tree, resulting in a total time and space complexity of $O(n_e\log{n_e})$.
\end{proof}

Applying the same technique as in Lemma~\ref{lemma:ADA}, we can determine the total query time and space complexity.

\begin{lemma}
	\label{lemma:RFS_1}
	The query time complexity of the Range Forest Solution is $O(\vert E \vert (T_{sp} + L \cdot \log\frac{N}{\vert E \vert}))$, while the construction time is at most $O(N \cdot \log{N})$.
\end{lemma}
\begin{proof}
	The query time on each edge is both $\log{n_e}$ for ADA and RFS. Therefore, they share the same time complexity. The constructing time on each edge $e$ is $O(n_e \log{n_e})$, so the worst case will cost $O(N \cdot \log{N})$ time to construct a range forest contains all events.
\end{proof}

\begin{algorithm}[t]
	\caption{Range Forest Construction}
	\label{algo:RFC}
	\DontPrintSemicolon
	\KwIn{event set $\{o_i\}$ on the edge $e$}
	\KwOut{range forest $T_i$}

	sort $\{o_i\}$ by time $t_i$
	
	\For{$i \in [1, n_e]$}{
		$u_l \leftarrow root(T_{i-1})$
	
		$root(T_i) \leftarrow DualConstruct(u_l, o_i)$
	}
		
	\SetKwFunction{DC}{DualConstruct}
	\SetKwProg{Fn}{Function}{:}{}
	\Fn{\DC($u_l, o_i$)}{
	
	create a new tree node $u_r$
	
	$\mathbf{A}(u_r) \leftarrow \mathbf{A}(u_l) + \mathbf{A}(o_i)$
		 
	\uIf{$d(v_c, p_i) \in R(lc(u_r))$}{
		link $rc(u_r)$ to $rc(u_l)$
		
		$lc(u_r) \leftarrow DualConstruct(lc(u_l), o_i)$
	}
	\Else{
		link $lc(u_r)$ to $lc(u_l)$
		
		$rc(u_r) \leftarrow DualConstruct(rc(u_l), o_i)$
	}
	
	\Return{$u_r$}
	}
		
\end{algorithm}

%% file: proposedSolution2.tex
\section{Streaming Update}
\label{sec7:dynamic}

	The Range Forest Solution achieves the same time complexity to support temporal queries. However, the structure of the range forest is fixed, requiring loading all events first to determine the structure initially.

	To support the streaming update, wherein events continuously arrive with the time order, it is necessary to replace the static range forest with a dynamic structure. Shortly, each tree node will maintain a specific segment of the edge containing an arbitrary number of events.

\subsection{Dynamic Range Forest Solution}

	The main structure of RFS is the range tree so we first illustrate how to make a range tree dynamic. Silimar to the static range tree, each dynamic range tree node $u$ also has a range $R(u)$. Now the range of the root node is $R(root) = [0, d(v_c, v_d)]$. Then, a tree node with range $R(u) = [l, r]$ will spilt into two child nodes where $R(lc(u)) = [l, (l + r) / 2]$ and $R(rc(u)) = [(l + r) / 2 + \delta, r]$. Here $\delta$ is a tiny value to separate two intervals. In short, the split of the range tree depends on the real but not relative positions.

	Figure \ref{fig:DRFS_1} shows an example of a dynamic range forest, where $\{o_1\}, \{o_2\}, \{o_3, o_4\}$ is located at the first, third, fourth quarter section of the edge, respectively. Algorithm~\ref{algo:RFQ} and Algorithm~\ref{algo:RFC} are also suitable for the query and construction of the dynamic range forest, respectively, with the range $R(u)$ modified.

\begin{figure}[h!]\centering
	\scalebox{0.32}[0.32]{\includegraphics{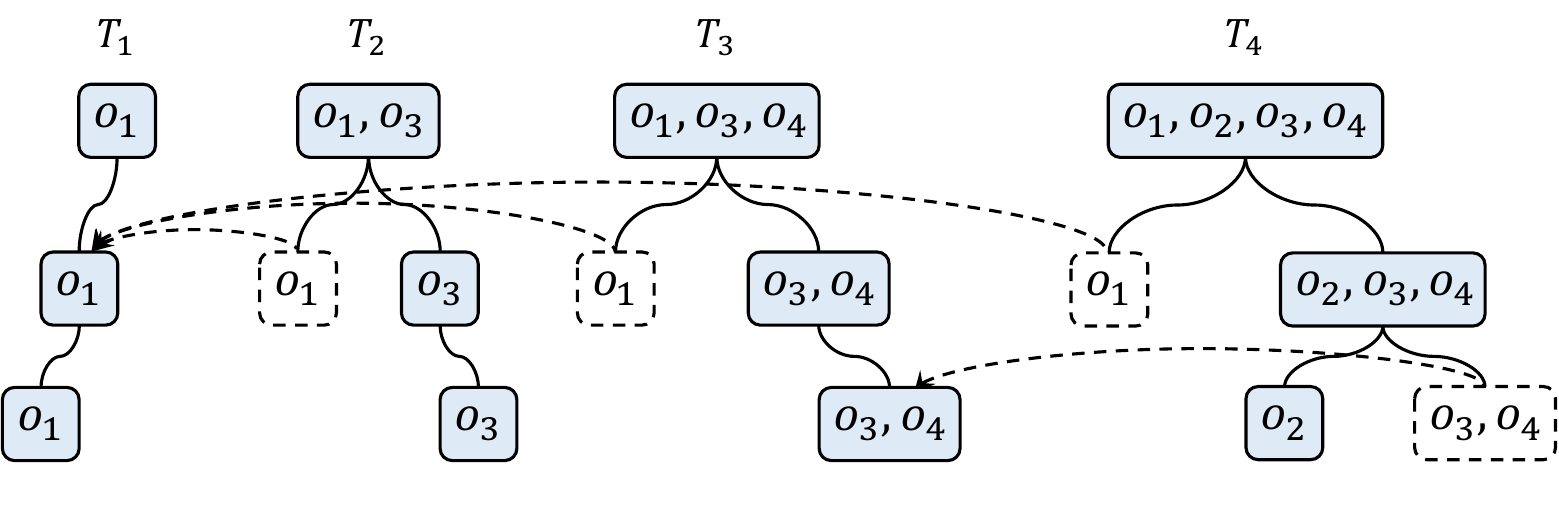}}
	\vspace{-2em}
	\caption{An example of dynamic range forest . Now the range forest is split based on the real position distance.}
	\vspace{-1em}
	\label{fig:DRFS_1}
\end{figure}

	The dynamic range forest may not be exact. In Figure~\ref{fig:RFS1}, each leaf node of the static range forest contains only one event. Therefore, RFS can always terminate at one leaf node and retrieve the exact aggregation of a query range. However, Figure~\ref{fig:DRFS_1} shows a leaf node with two event $\{o_3, o_4\}$. If this node is partially covered, the process will try to recurse to a deeper layer, which is of course an illegal operation.

	To overcome this problem, we propose the extension operation, which can extend the layer of the forest. For example, the $3$-layer structure is not enough for the situation in Figure~\ref{fig:DRFS_1}, and we try to extend the fourth layer from each leaf node. Figure~\ref{fig:DRFS_2} is the $4$-layer structure and the extended nodes are bounded by a rectangle.

\begin{figure}[h!]\centering
	\scalebox{0.30}[0.30]{\includegraphics{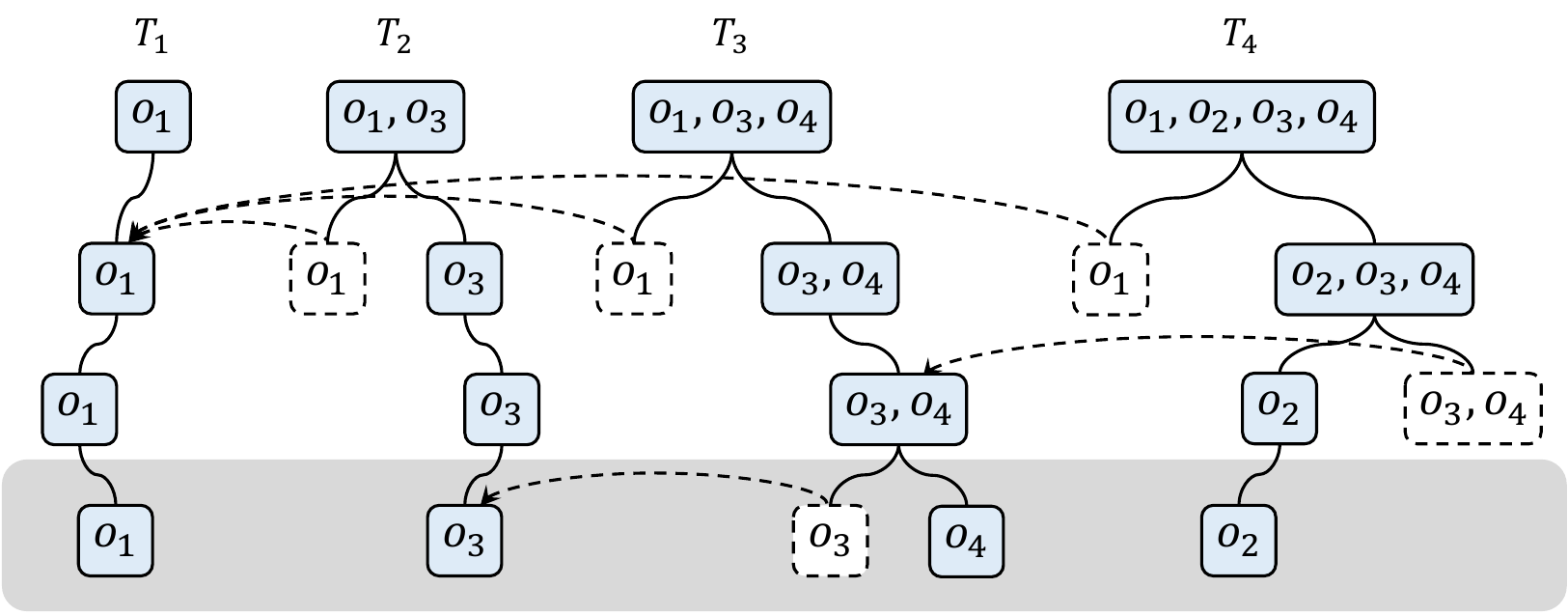}}
	\vspace{-1em}
	\caption{An example of dynamic range forest. The extended level is bounded by a rectangle.}
	\vspace{-1em}
	\label{fig:DRFS_2}
\end{figure}

	The extension operation can be regarded as a lazy evaluation and the implementation is shown in Algorithm~\ref{algo:DRFS}. We save the leaf node $leaf_i$ from the last construction. When expanding, the \textit{DualConstruct} function can be executed in continuing on each leaf node to extend a deeper layer. To recover the constructing process, only two stack variables $u_l$ and $o_i$ are required to be recorded.
	
	Therefore, we can pre-construct the range forest with a default depth $H$, forming a $H$-layer range forest. Then user can dynamically adjust the parameter $H$ based on their accuracy requirement and time limitation. Specifically, a larger $H$ will lead to higher accuracy but more time cost. Compared with constructing a $H$-layer range forest directly, this gradual construction is more flexible and will not add extra cost.

	\begin{lemma}
		The query time complexity of the Dynamic Range Forest Solution (DRFS) is $O(\vert E \vert (T_{sp} + L \cdot H))$ where the forest depth is $H$. The construction time and space complexity is both $O(N \cdot H)$ whether constructed directly or dynamically.
	\end{lemma}
	\begin{proof}
		Following Lemma~\ref{lemma:RFS_1}, we replace the tree depth $\log{n_e}$ with $H$ to get the query time complexity. Constructing an $H$-layer dynamic range forest directly costs $O(n_e \cdot H)$, while extending a new layer using Algorithm \ref{algo:DRFS} takes $O(n_e)$. Since $O(n_e \cdot H) + O(n_e) = O(n_e \cdot (H{+}1))$, Algorithm \ref{algo:DRFS} does not incur any extra cost and the total construction time and space is both $O(N \cdot H)$.
	\end{proof}

\subsection{Quantization}

	The index size is highly sensitive for real-world applications. A relatively smaller index is preferable as long as the accuracy is acceptable. An observation is that the query time $O(H)$ is a strict bound. In practice, the query process often terminates earlier when no tree node is partially covered. Thereby, a relatively small $H$ can yield sufficiently accurate results. This compressing operation is referred to as quantization.

	The implementation of quantization involves strategically \textit{hiding} some deeper layers. Given a quantized depth $H_0 < H$, the query process is designed to always terminate at the $H_0$-th layer. In cases where the terminated node is partially covered, the returned value is a zero-vector.

	The accuracy of quantization is directly determined by the quantized depth $H_0$. A reduction in $H_0$ can significantly decrease both time and space consumption. In exchange, leaf nodes will contain more events, leading to more partially covered cases potentially. In our experiments, even $H_0=2$ achieves more than $90\%$ accuracy compared with the exact results and costs only about $2^{H_0}=4$ times the memory of the original data, revealing the high scaling ability of the quantization.

\begin{algorithm}[t]
	\caption{Extension Operation}
	\label{algo:DRFS}
	\DontPrintSemicolon
	\KwIn{a range forest, event set $\{o_i\}$}
	\KwOut{a range forest with an extended layer}
	
	\For{each event $o_i$}{

		$u \leftarrow$ last updated node $leaf_i$
		
		call DualConstruct on $u$ to extend a new layer

		save new updated node $leaf_i$
	}

\end{algorithm}

%% file: proposedSolution3.tex
\section{Lixels Sharing}
\label{sec6:LixelsAggregaion}

The range forest, as well as its dynamic variant, exhibits an efficient performance on a single query, while multiple queries among different lixels are still independent, which is another bottleneck of these KDE-based methods. In this section, we introduce an optimization specialized for the polynomial kernel functions called \textit{Lixels Sharing} that can share density values between different lixels.

The massive computation originates from the various query ranges and frequent query processes on the index. Intuitively, if the query range is sufficiently large~(such as the example in Figure~\ref{fig:LS_1}), the aggregation will include all events at the root node of the range forest.

\begin{figure}[h]\centering
    \scalebox{0.6}[0.6]{\includegraphics{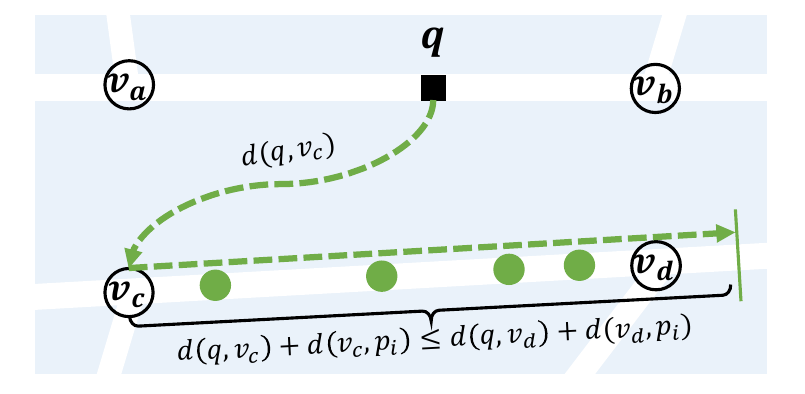}}
    \centering
    \vspace{-1em}
    \caption{An example of dominated edges. The shortest paths from $q$ to all events will throught $v_c$.}
    \label{fig:LS_1}
\end{figure}

\subsection{Domination Determination}

If all lixels on the edge $(v_a, v_b)$ follow this situation, i.e., the query range covers the whole edge $(v_c, v_d)$, we call it \textit{dominated}. For example, edge $(v_c, v_d)$ is dominated at $v_c$ in Figure~\ref{fig:LS_1} since the shortest paths from all lixels $q_i$ to all events will through $v_c$. Formally, for all lixels $q_i \in (v_a, v_b)$ and events $p_i \in (v_c, v_d)$, the domination has two conditions:
\begin{equation}
\label{eq:LS_1}
\begin{aligned}
    d(q_i, v_c) + d(v_c, p_i) &\le b_s \\
    d(q_i, v_c) + d(v_c, p_i) &\le d(q_i, v_d) + d(v_d, p_i) \\
\end{aligned}
\end{equation}

The first condition claims all events are within the spatial bandwidth. Considering the worse case, the maximum $d(v_c, p_i)$ is obviously the whole edge length $d(v_c, v_d)$, while the maximum $d(q_i, v_c)$ is half of the length of the loop $v_c {\rightarrow} v_a {\rightarrow} v_b {\rightarrow} v_c$, i.e., $C / 2$ where $C = d(v_c, v_a) + d(v_a, v_b) + d(v_b, v_c)$ is the loop length.

The second condition claims all events are closer to $v_c$ than $v_d$, which is equivalent to:
\begin{equation*}
    \max_{q_i}\{d(q_i, v_c) - d(q_i, v_d)\} \le \min_{p_i}\{d(v_d, p_i) - d(v_c, p_i)\}
\end{equation*} 
The right-hand side is constant since the expression has monotonicity and $p_{n_e}$ achieves the minimum value. However, the left part is unaffordable to compute one-by-one. Actually, we have the following observation:

\begin{lemma}
    \label{lemma:LS_1}
    The maximum of $d(q_i, v_c) - d(q_i, v_d)$ appears at $4$ positions at most.
\end{lemma}
\begin{proof}
    Shown in Figure~\ref{fig:LS_2}, The route of $d(q_i, v_c)$ passes $v_a$ first and then passes $v_b$. Therefore, the values of $d(q_i, v_c)$ is always increasing from $q_1$ to $q_k$, and then decreasing from $q_{k+1}$ to $q_{l_e}$, where $k$ is a break point and $l_e$ is the number of lixels on the edge $(v_a, v_b)$. The common difference is $g$ and $-g$, respectively.

    The $v_d$ side has a similar situation. The values of $d(q_i, v_d)$ is increasing from $q_1$ to $q_{k'}$, and then decreasing from $q_{k'+1}$ to $q_{l_e}$, where $k'$ is another break point. The common difference is also $g$ and $-g$, respectively.

    Now $d(q_i, v_c) - d(q_i, v_d)$ is the subtraction of two arithmetic sequences. From $q_1$ to $q_{\min\{k,k'\}}$ and from $q_{\max\{k,k'\}+1}$ to $q_{l_e}$, the common difference is both $g$ and $-g$, respectively. As a result, $d(q_i, v_c) - d(q_i, v_d)$ is constant. From $q_{\min\{k,k'\}+1}$ to $q_{\max\{k,k'\}}$, the common difference is reversed, so the result is also an arithmetic sequence, where the maximum value is one of the endpoints.

    Finally, at most $4$ positions are required to check, $k$, $k+1$, $k'$, and $k'+1$.
\end{proof}

\subsection{Domination Computation}

Using Lemma \ref{lemma:LS_1}, we can quickly determine whether an edge is dominated. If an edge is dominated, the aggregated vector $\mathbf{A}$ in Equation (\ref{eq:aggregate}) will be the same for all lixel queries. Thus, the only variable in Equation (\ref{eq:aggregate}) is $d(q, v_c)$. Due to the Triangular kernel function, $d(q, v_c)$ is a linear variable, i.e.,
\begin{equation*}
    F_e(q_i) = \alpha \cdot d(q_i, v_c) + \beta
\end{equation*}
Moreover, the density value difference between two adjacent lixels $q_i$ and $q_{i-1}$ is:
\begin{equation*}
    F_e(q_i) - F_e(q_{i-1}) = \alpha \cdot (d(q_i, v_c) - d(q_{i-1}, v_c)) = \alpha \cdot \Delta(q_i)
\end{equation*}

\begin{figure}[t]\centering
    \scalebox{0.6}[0.6]{\includegraphics{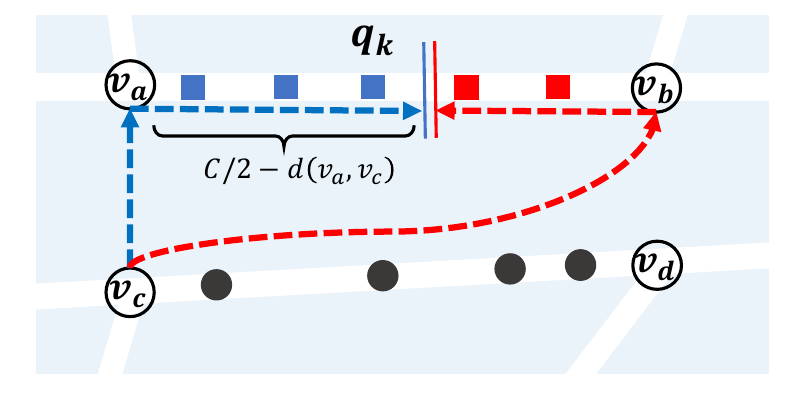}}
    \centering
    \vspace{-1em}
    \caption{An example of dominated edges. The shortest paths from $v_c$ to all lixels have two possible routes~(passing $v_a$ or $v_b$).}
    \label{fig:LS_2}
\end{figure}

Lemma \ref{lemma:LS_1} has illustrated that the $d(q_i, v_c)$ consists of two arithmetic sequences breaking at $k$. Then the first-order difference $\Delta(q_i)$ consists of two constant sequences, and the second-order difference $\Delta^2(q_i)$ are all zero. Figure~\ref{fig:LS_3} visualizes this idea. The density values $F_e(q_i)$ can be replaced by two updates on the second-order difference.

\begin{figure}[h!]\centering
    \scalebox{0.39}[0.39]{\includegraphics{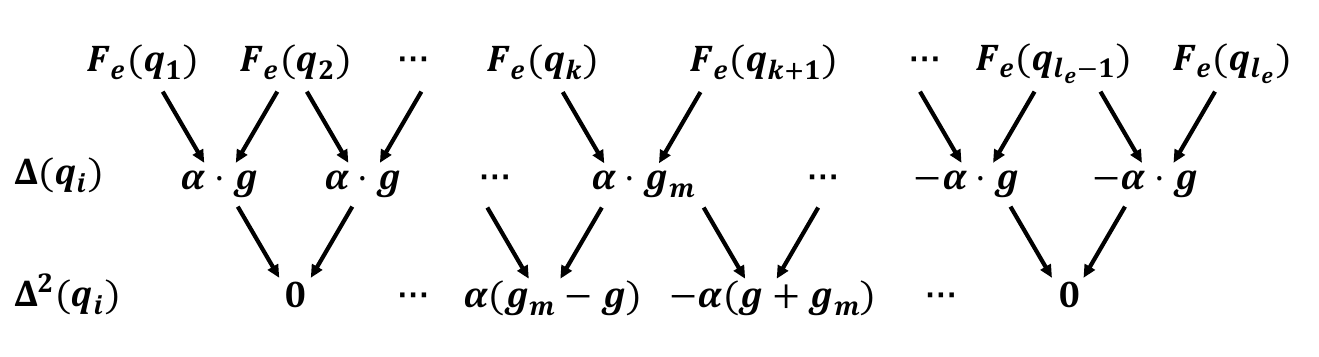}}
    \caption{ The first and the second-order difference $\Delta(q_i)$ and $\Delta^2(q_i)$, respectively. $\Delta^2(q_i)$ are all zero except two items near the middle point.}
    \label{fig:LS_3}
\end{figure}

For different dominated edges, we can easily update more values over $\Delta^2(q_i)$. Then following the rule that the differences of summation equals the summation of differences, these density values can be recovered from $\Delta(q_i)$ and $\Delta^2(q_i)$.

Finally, we present how to locate the break point $k$. As shown in Figure~\ref{fig:LS_2}, $q_k$ is the farthest lixel that passing $v_a$ and we have:
\begin{equation*}
    k = \left\lceil \frac{C/2 - d(v_a, v_c)}{g} + 0.5 \right\rceil
\end{equation*}
Here the additional $0.5$ means the middle point of the lixel $q_k$ must be within the boundary. Once $k$ is determined, we have:
\begin{equation*}
    g_m = d(q_{k+1}, v_c) - d(q_k, v_c)
\end{equation*}

\subsection{Out-of-the-bandwidth Determination}

The example in Figure~\ref{fig:LS_1} shows that if the aggregation contains all events, we can optimize this process. On the contrary, if the aggregation contains no event, we can also skip this edge for all lixels. Specifically, we have the following determining condition:
\begin{equation*}
\begin{aligned}
    d(q_i, v_c) + d(v_c, p_i) > b_s \\
    d(q_i, v_d) + d(v_d, p_i) > b_s
\end{aligned}
\end{equation*}
i.e., all events are out of the spatial bandwidth. Considering the worse case, $d(v_c, p_i) = d(v_d, p_i) = 0$. The minimum value of $d(q_i, v_c)$ and $d(q_i, v_d)$ only appears at the endpoint $q_1$ or $q_{l_e}$ since they are always increasing first and then decreasing according to Lemma~\ref{lemma:LS_1}. Therefore, if these two lixels are too far, the edge $(v_c, v_d)$ can be directly skipped.

\subsection{Lixel Sharing Framework}

Algorithm \ref{algo:LixelSharing} is the optimized framework using the Lixel Sharing. First, Lines 3-5 finds the dominated edge set $E_d$ and out-of-bandwidth edge set $E_o$, while the remaining edges are saved in $E_q$. For dominated edges, update the second order difference $\Delta^2(q_i)$ in lines 6-7 and recover the density values in line 8. Lines 9-12 will compute density values for the remaining edges with original query processing. 

\begin{algorithm}[t]
	\caption{Lixel Sharing Framework}
	\label{algo:LixelSharing}
	\DontPrintSemicolon
	\SetKwComment{comment}{$\triangleright$ }{}
	
    \For{each edge $(v_a, v_b)$}{
        Get shared shortest path distance \\
        $E_d \leftarrow$ dominated edges \\
        $E_o \leftarrow$ out-of-bandwidth edges \\
        $E_q \leftarrow E \backslash (E_d \cup E_o$) \\
        \For{each edge $e \in E_d$}{
            update $\Delta^2(q_i)$
        }
        recover $F(q_i)$ from $\Delta(q_i)$ and $\Delta^2(q_i)$
        
        \For{each lixel $q_i \in (v_a, v_b)$}{
            \For{each edge $e \in E_q$}{
                compute $F_e(q_i)$ \\
                $F(q_i) \leftarrow F(q_i) + F_e(q_i)$
            }
	    }
    }
\end{algorithm}

\begin{lemma}
    The Lixel Sharing Framework shown in Algorithm \ref{algo:LixelSharing} cost $O(\vert E \vert \cdot T_{sp} + \vert E \vert^2 + L \cdot \vert E_q \vert \cdot \log \frac{N}{E_q})$ time complexity.
\end{lemma}
\begin{proof}
    Line 2 costs $O(\vert E \vert \cdot T_{sp})$ in shortest path algorithm and lines 7-10 costs $O(L \cdot (\vert E_q \vert \log \frac{N}{\vert E_q \vert}))$ using RFS, which is unchanged. Determining the domination and the bandwidth of all edges in Lines 3-4 will cost $O(\vert E \vert)$. The loop in lines 6-7 will cost $O(\vert E_d \vert)$. The recover operation in line 8 will cost $O(L)$ in total. Absorbing $O(L)$ item, the total time complexity is $O(\vert E \vert \cdot T_{sp} + \vert E \vert^2 + L \cdot \vert E_q \vert \cdot \log \frac{N}{E_q})$.
\end{proof}

The efficiency of Algorithm \ref{algo:LixelSharing} is determined by the size of $E_q$. In the worst case, $E_q = E$ and Algorithm \ref{algo:LixelSharing} has no improvement. As the reduction of $E_q$, the time complexity will rapidly decline.

%% file: proposedSolution4.tex
\section{Non-Polynomial Kernel Functions}
\label{sec7:kernel}

As an example kernel function, the Triangular kernel function is widely applied in KDE-based problems since it can be divided into a query vector $\mathbf{Q}$ and an aggregated vector $\mathbf{A}$ shown in Equation (\ref{eq:aggregate}), as well as other polynomial kernel functions like the Epanechnikov kernel function. However, many non-polynomial kernel functions, such as the Exponential kernel function and the Cosine kernel function, do not have this feature. 

Recent works try to use a polynomial kernel function to approach the non-polynomial kernel function, such as two Linear functions \cite{chan_karl_2019} and the Quadratic function \cite{chan_quad_2020}. These methods achieve high accuracy and tight boundaries. However, they are not generic (1) because the parameters need to be carefully considered, and (2) the error always exists and will not converge.

To improve these issues, our framework can support many non-polynomial kernel functions by decomposing the density value into the production of a query vector and an aggregated vector:
\begin{equation}
	F_\Gamma(q)= \mathbf{Q}(q) \cdot \mathbf{A}(\Gamma)
\end{equation}

\subsection{Exponential Kernel Function}

We first present the division of the Exponential kernel function:
\begin{equation}
	K(x) = e^{-x}\notag
\end{equation}

To simplify the problem, We only apply it to the spatial kernel function, i.e., $K_s(x) = e^{-x}$ and $K_t(x) = 1$. Therefore, the variable $x$ in $K_s(x)$ equals to $d(q, p_i)/b_s$ and the density value will be:
\begin{equation}
\begin{aligned}
	F_{\Gamma}(q) = \sum_{o_i \in \Gamma} e^{-\frac{d(q, v_c) + d(v_c, p_i)}{b_s}} 
	= e^{-d(q, v_c)/b_s} \sum_{o_i \in \Gamma} e^{-d(v_c, p_i)/b_s} \notag
\end{aligned}
\end{equation}

Actually, the query vector and the aggregated vector have only one item:
\begin{equation*}
	\mathbf{Q}(q) = e^{-d(q, v_c)/b_s}
	\quad
	\mathbf{A}(\Gamma) = \sum_{o_i \in \Gamma} e^{-d(v_c, p_i)/b_s}
\end{equation*} 

\subsection{Cosine Function}

Another widely used kernel function is the Trigonometric kernel function, such as the Cosine kernel function:
\begin{equation}
	K(x) = \cos(x)\notag
\end{equation}

Since the Cosine function cannot be directly divided, the trigonometric rules are applied:
\begin{equation}
	\begin{aligned}
		&F_\Gamma(q) \\
		=& \sum_{o_i \in \Gamma} \cos\left(d(q, v_c)/b_s + d(v_c, p_i)/b_s\right) \\
		=& \sum_{o_i \in \Gamma} \left[\cos \frac{d(q, v_c)}{b_s} \cos \frac{d(v_c, p_i)}{b_s} - \sin \frac{d(q, v_c)}{b_s} \sin \frac{d(v_c, p_i)}{b_s}\right] \\
		=& \cos \frac{d(q, v_c)}{b_s} \sum_{o_i \in \Gamma} \cos \frac{d(v_c, p_i)}{b_s} 
		-  \sin \frac{d(q, v_c)}{b_s} \sum_{o_i \in \Gamma} \sin \frac{d(v_c, p_i)}{b_s}
	\end{aligned}\notag
\end{equation}

Therefore, the query vector and the aggregated vector both have two items:
\begin{equation*}
	\mathbf{Q}(q) = 
	\begin{bmatrix}
		\cos \frac{d(q, v_c)}{b_s} \\
		-\sin \frac{d(q, v_c)}{b_s}
	\end{bmatrix}^\top
	\quad
	\mathbf{A}(\Gamma) = 
	\begin{bmatrix}
		\sum_{o_i \in \Gamma} \cos \frac{d(v_c, p_i)}{b_s} \\
		\sum_{o_i \in \Gamma} \sin \frac{d(v_c, p_i)}{b_s}
	\end{bmatrix}
\end{equation*}

\subsection{Multi-Kernel Function}

Previously, we only applied these non-polynomial kernel functions on $K_s(x)$ and $K_t(x)$ remains constant. Moreover, $K_s(x)$ and $K_t(x)$ can be the combination of arbitrary functions mentioned above:
\begin{equation}
\begin{aligned}
	F_\Gamma(q) &= \Big[\mathbf{Q}_s(q) \cdot \mathbf{A}_s(\Gamma)\Big] \cdot \Big[\mathbf{Q}_t(q) \cdot \mathbf{A}_t(\Gamma)\Big] \\
	&=\left(\sum_i \mathbf{Q}_i(q) \cdot \mathbf{A}_i(\Gamma)\right) \cdot \left(\sum_j \mathbf{Q}_j(q) \cdot \mathbf{A}_j(\Gamma)\right) \\
	&=\sum_{i,j} \mathbf{Q}_{ij}(q) \cdot \mathbf{A}_{ij}(\Gamma)
\end{aligned}
\end{equation}
where $\mathbf{Q}_{ij}(q)=\mathbf{Q}_i(q) \cdot \mathbf{Q}_j(q)$ and $\mathbf{A}_{ij}(\Gamma)=\mathbf{A}_i(\Gamma) \cdot \mathbf{A}_j(\Gamma)$. Therefore, users can define the spatial and the temporal kernel functions separately and arbitrarily. The size of the final aggregated vector $\vert \mathbf{A}_{ij} \vert$ is up to $\vert \mathbf{A}_i \vert \cdot \vert \mathbf{A}_j \vert$, which is still $O(1)$ in practice.

%
%
%


%% file: experimentalStudy.tex
\begin{figure*}[t!]\centering
\hspace{-8pt}
\scalebox{0.3}[0.3]{\includegraphics{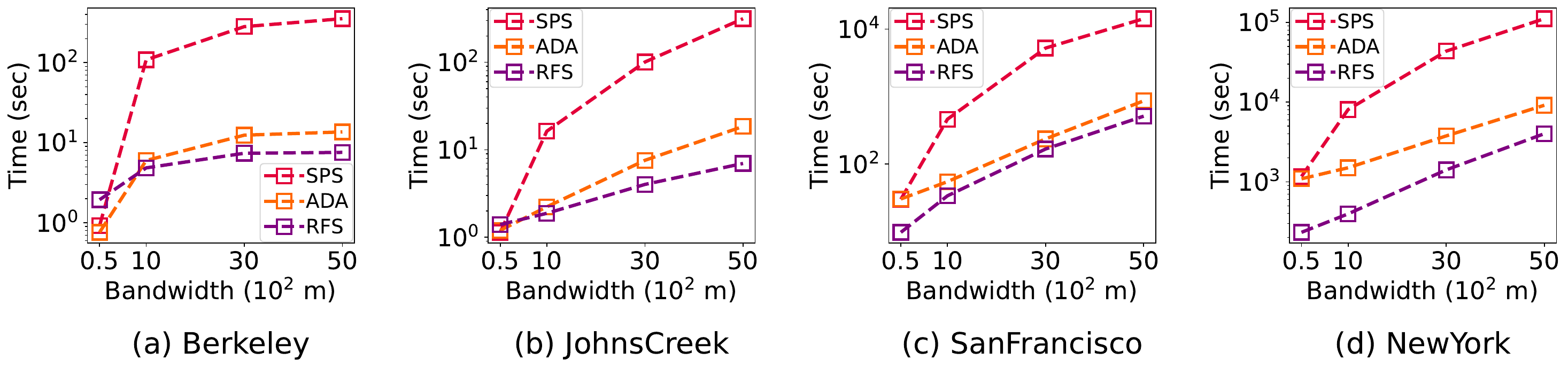}}
\vspace{-8pt}
\caption{Processing time with different spatial bandwidths, varying in 50m, 1000m, 3000m, 5000m.}
\label{exp1.1}

\vspace{10pt}
\scalebox{0.3}[0.3]{\includegraphics{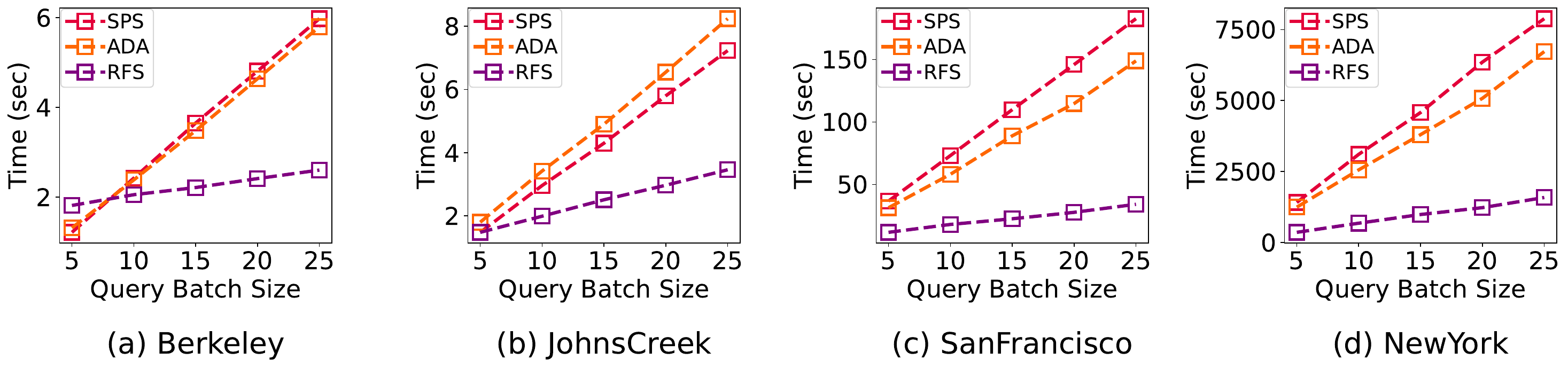}}
\vspace{-8pt}
\caption{Processing time with different query batch sizes, varying in 5, 10, 15, 20, 25 queries.}
\label{exp1.2}

\vspace{10pt}
\scalebox{0.3}[0.3]{\includegraphics{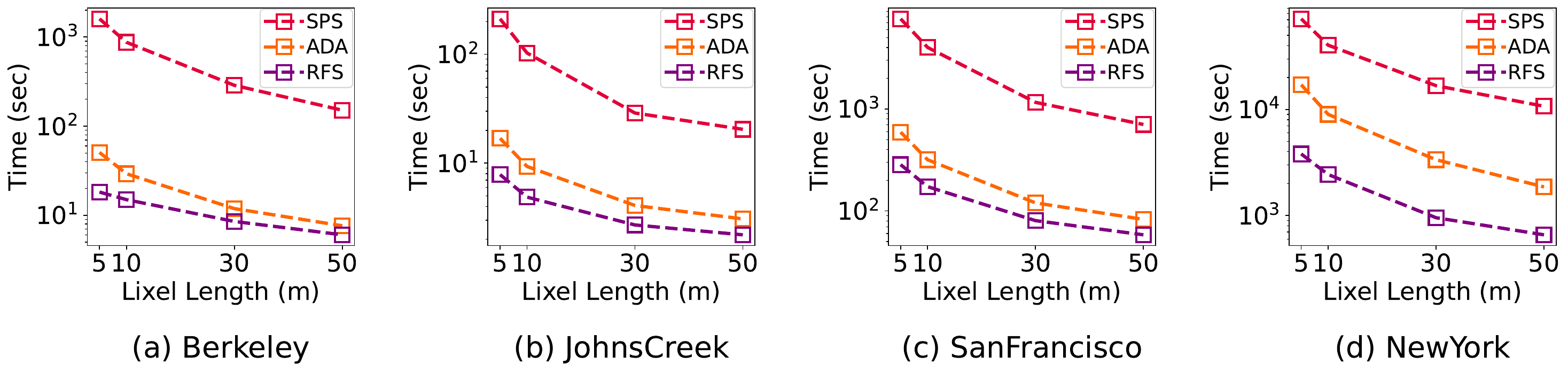}}
\vspace{-8pt}
\caption{Processing time with different lixel lengths, varying in 5m, 10m, 30m, 50m.}
\label{exp1.3}

\vspace{10pt}
\scalebox{0.3}[0.3]{\includegraphics{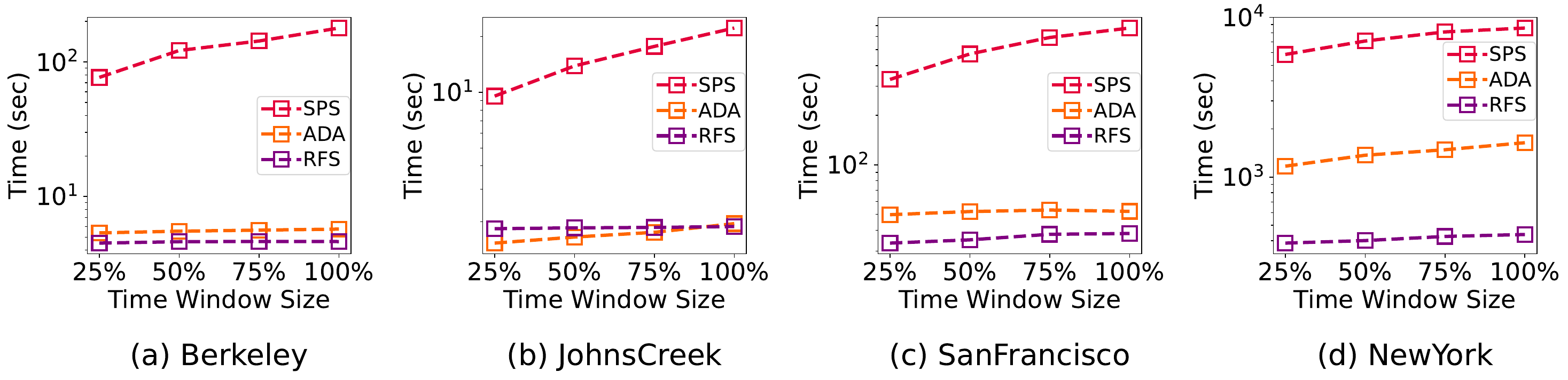}}
\vspace{-8pt}
\caption{Processing time with different time window sizes, varying in 25\%, 50\%, 75\%, 100\%.}
\label{exp1.4}
\end{figure*}

\section{Experiment}
\label{sec8:exp}

This section evaluates the efficiency of our proposed algorithms on various datasets with different parameters. We implemented all algorithms in C++17 and compiled them with g++ at the -O3 optimization level. Our code was run on an Intel Xeon Gold 6258R CPU @ 2.70GHz with 128GB memory. We used the Triangular Kernel Function as the default kernel function and Dijkstra as the shortest path distance computation algorithm in all algorithms.

\subsection{Datasets}

We used four datasets with various scales and categories, as shown in Table \ref{tab:datasets}. All networks were downloaded from OpenStreetMap using OSMnx. We assumed that each edge is a straight line and each event was matched to the nearest edge. The event datasets include police calls, paid parking records, and taxi trips. 

\begin{table}[h]
\centering
\caption{Parameters of datasets.}
\vspace{-5pt}
\label{tab:datasets}
\begin{tabular}{c|cccc} 
\hline
Dataset     & $\vert V \vert$ & $\vert E \vert$ & $N$ & $N / \vert E \vert$  \\ 
\hline
Berkeley    & 1576            & 4378            & 735366          & 168                              \\ 
\hline
Johns Creek & 3074            & 3471            & 979072          & 282                              \\ 
\hline
San Francisco  & 9700         & 16008           & 5379023         & 336                              \\ 
\hline
New York	& 55765           & 92229           & 38400730        & 416                              \\ 
\hline
\end{tabular}
\end{table}

\subsection{Compared with Pervious Work}

First, we evaluate the efficiency of our proposed solution, Range Forest Solution (RFS) with the Lixel Sharing optimization, in generating exact KDE values compared with two previous works, Shortest Path Sharing (SPS) and Aggregate Distance Augmentation (ADA). SPS uses the basic shortest path sharing framework with no index, while ADA will filter events that within the time window and build a linear index by their distances.

Each query batch may include multiple queries with different time windows, which are given online in random order. Therefore, other time windows are invisible when processing one time window. By default, each time window includes $70\%$ events.

\textbf{Varying Bandwidth.} First, we evaluate the processing time with different bandwidths (50m, 1000m, 3000m, 5000m) with the single query in four datasets. The lixel length is 10m, which is an appropriate resolution.

The results are presented in logarithmic form in Figure \ref{exp1.1}. The index-free method SPS is slower than other index-based methods by 1$\sim$2 orders of magnitude, especially with a large bandwidth. RFS achieves at most 6 times speedup compared with ADA, but performs poorly in small networks with a lower bandwidth, costing too much time in indexing.

\textbf{Varying Query Batch Size.} The previous experiment uses only one query. However, the off-line method ADA has to re-index for a new query time window. Therefore, we measure the processing time with different query batch sizes varying in 5, 10, 15, 20, 25 query time windows. The spatial bandwidth is 50m and the lixel length is 50m. In this scenario, re-loading and re-indexing all events will cost too much time.

\begin{figure}[t]\centering
	\scalebox{0.4}[0.4]{\includegraphics{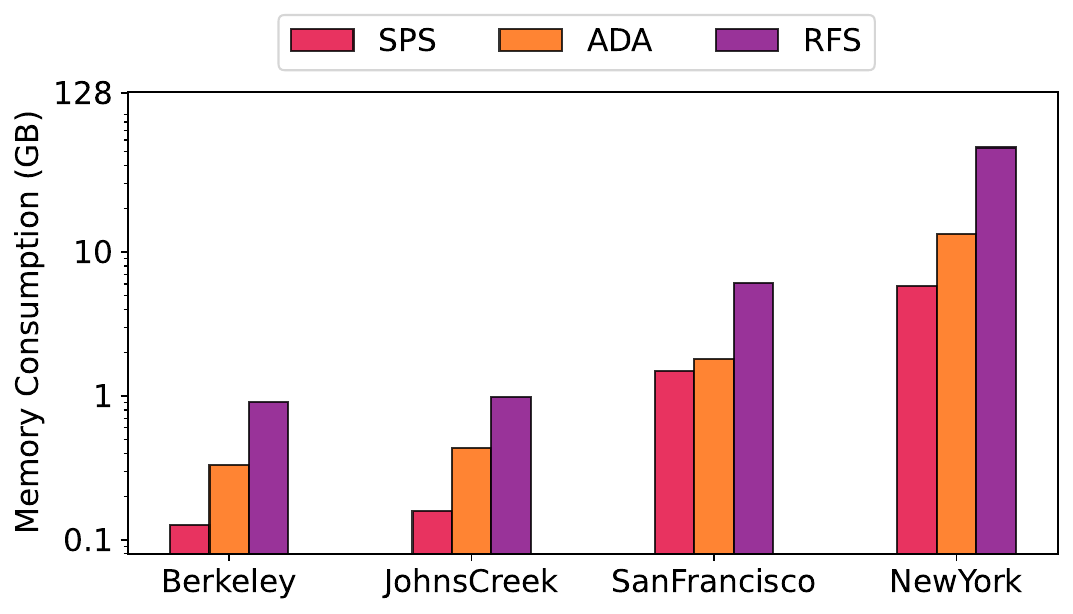}}
	\vspace{-8pt}
	\caption{Memory consumption of different methods.}
	\label{exp1.5}
\end{figure}

Figure \ref{exp1.2} shows a linear growth trend for all methods. The intercept is the indexing time while the slope is the processing time in each query. The preprocessing time of RFS is relatively higher in small datasets (i.e., Berkeley), but the growth ratio is much lower than SPS and ADA, leading to less processing time. For the other three datasets, RFS performs better in all cases.

\textbf{Varying Lixel Length.} The lixel length directly decides the resolution. For a sketchy quick view, 50m is enough since the average edge length is only 100m$\sim$200m. Users may further improve the resolution and set a lower lixel length. In this part, the lixel length varies in 5m, 10, 30m, 50m. Each query batch includes 5 time windows.

Figure \ref{exp1.3} displays the processing time in logarithmic form for different lixel lengths (5m, 10m, 30m, 50m). All results are roughly in the inverse proportion since a smaller lixel length leads to more lixels. RFS shows a significant advantage compared with other methods, especially in a lower lixel length (i.e., up to 4.5x speedup compared with ADA).

\textbf{Varying Time Window Size.} We also want to know how the event scale impacts and set different sizes of time windows varying in 25\%, 50\%, 75\%, 100\%.

As shown in Figure \ref{exp1.4}, SPS is linearly correlated to the number of events, so the processing time is increasing by the time window size. On the contrary, RFS has no relationship to the time window size. ADA is a little different since the time window size only affects the preprocessing step and the growth trend is not evident. Therefore, the efficiency of RFS will not be affected even with a larger number of events.

\textbf{Memory.} Figure \ref{exp1.5} shows the memory usage of three methods. Since SPS does not require any extra space, it can be seen as the dataset size. The memory usage of the other two methods is only related to the event size. RFS only cost 3 and 8 times the memory of ADA and SPS, respectively, to build the forest index. Considering the complex index structure, this memory cost is acceptable.

\begin{figure}[t!]\centering\vspace{-2ex}
	\scalebox{0.3}[0.3]{\includegraphics{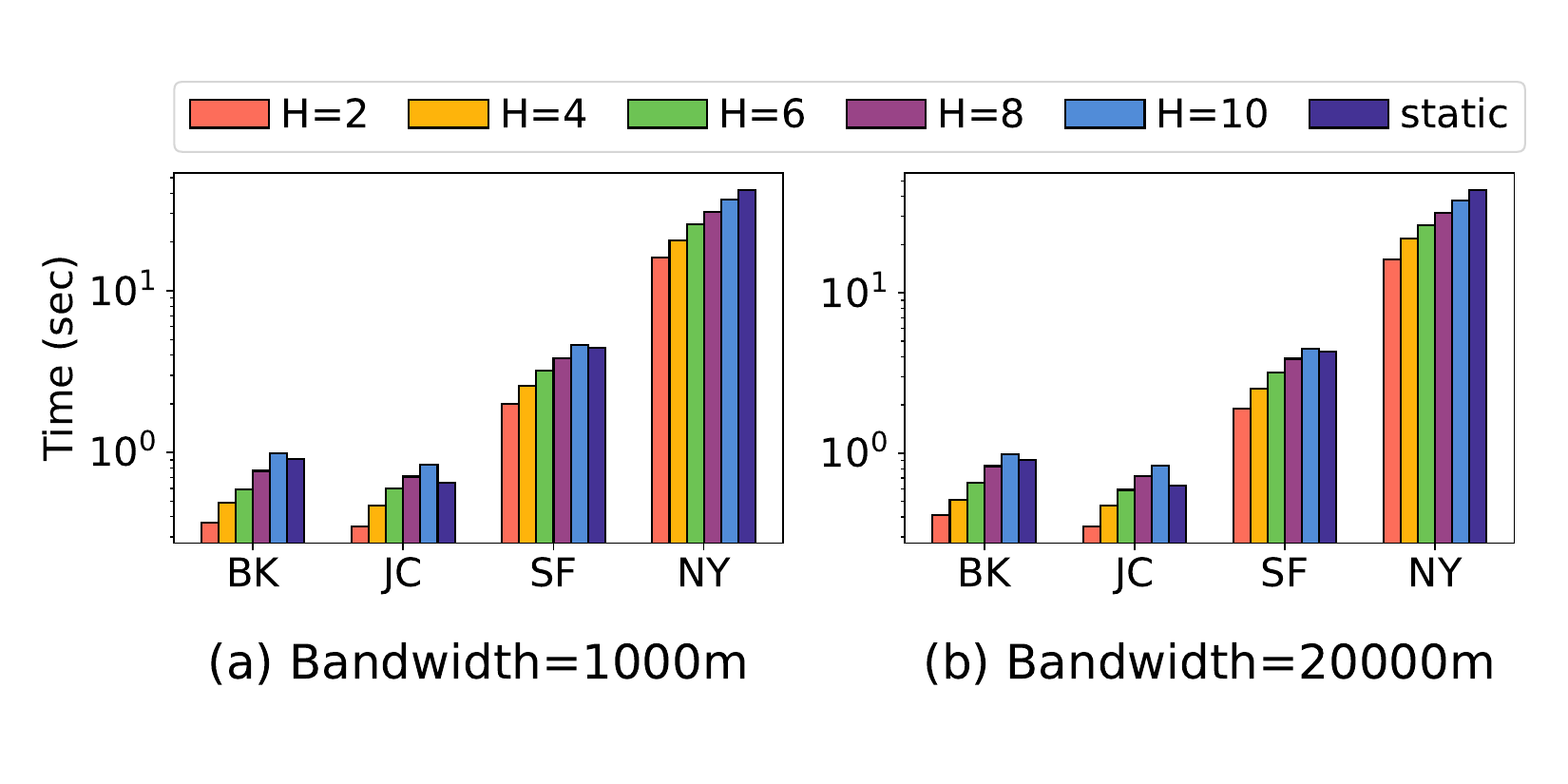}}
	\vspace{-20pt}
	\caption{Indexing time for different depth $H$.}
	\label{exp2.1}\vspace{-2ex}
\end{figure}

\subsection{DRFS Efficiency and Effectiveness}

This section evaluates the efficiency and accuracy of Dynamic Range Forest Solution (DRFS) with different depths $H$ ranging from 2 to 10. For comparison, we use RFS without the Lixel Sharing optimization, which has a static structure. We choose two bandwidths, 1000m and 20000m. The lixel length is fixed at 50m and the time window contains all events. The algorithm initially sets $H=1$ and gradually increases it to show the dynamic process.

\begin{figure}[t!]\centering
\scalebox{0.3}[0.3]{\includegraphics{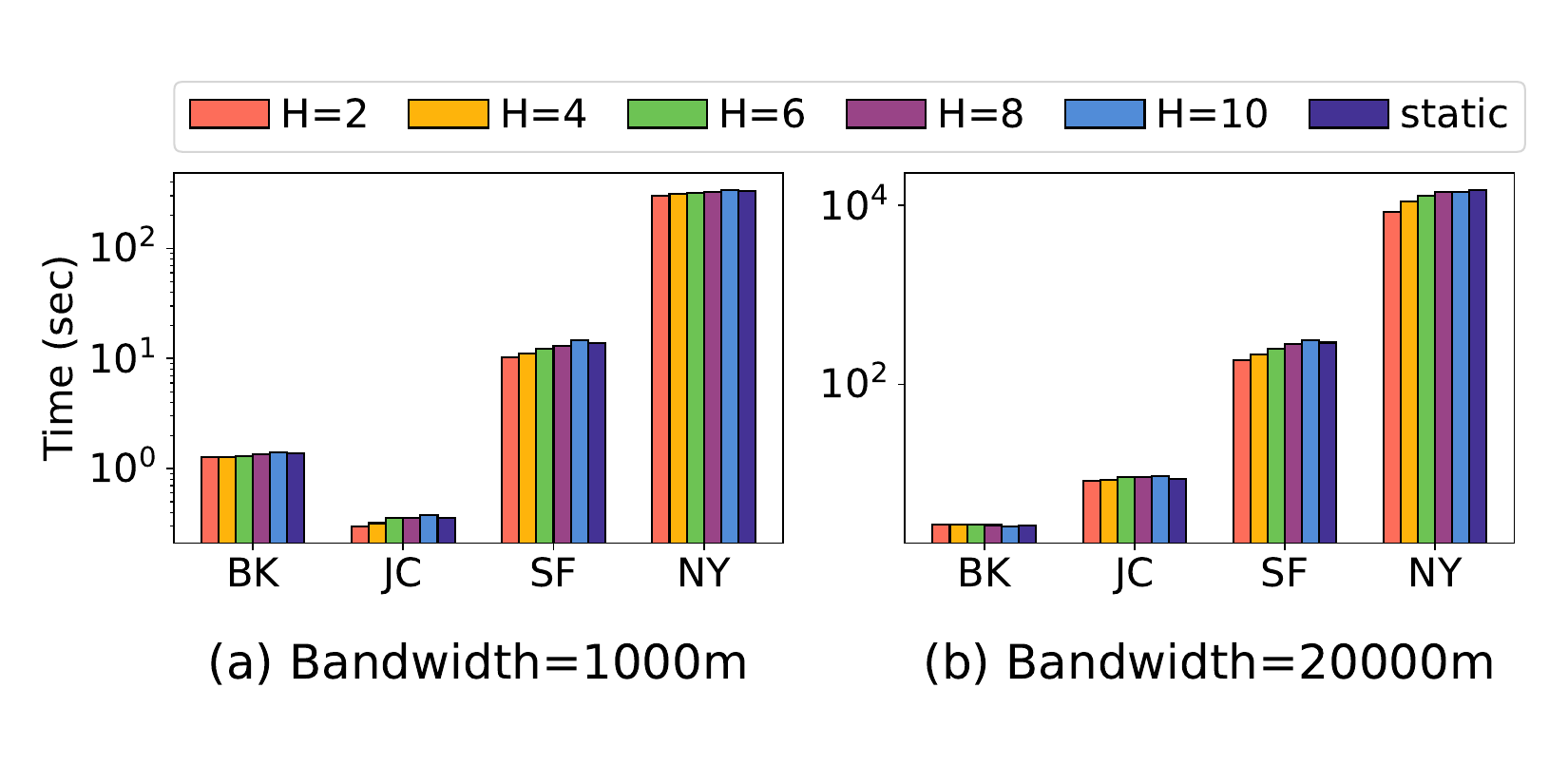}}
\vspace{-20pt}
\caption{Processing time for different depth $H$.}
\label{exp2.2}\vspace{-1ex}
\end{figure}

\begin{figure}[t!]\centering
	\scalebox{0.3}[0.3]{\includegraphics{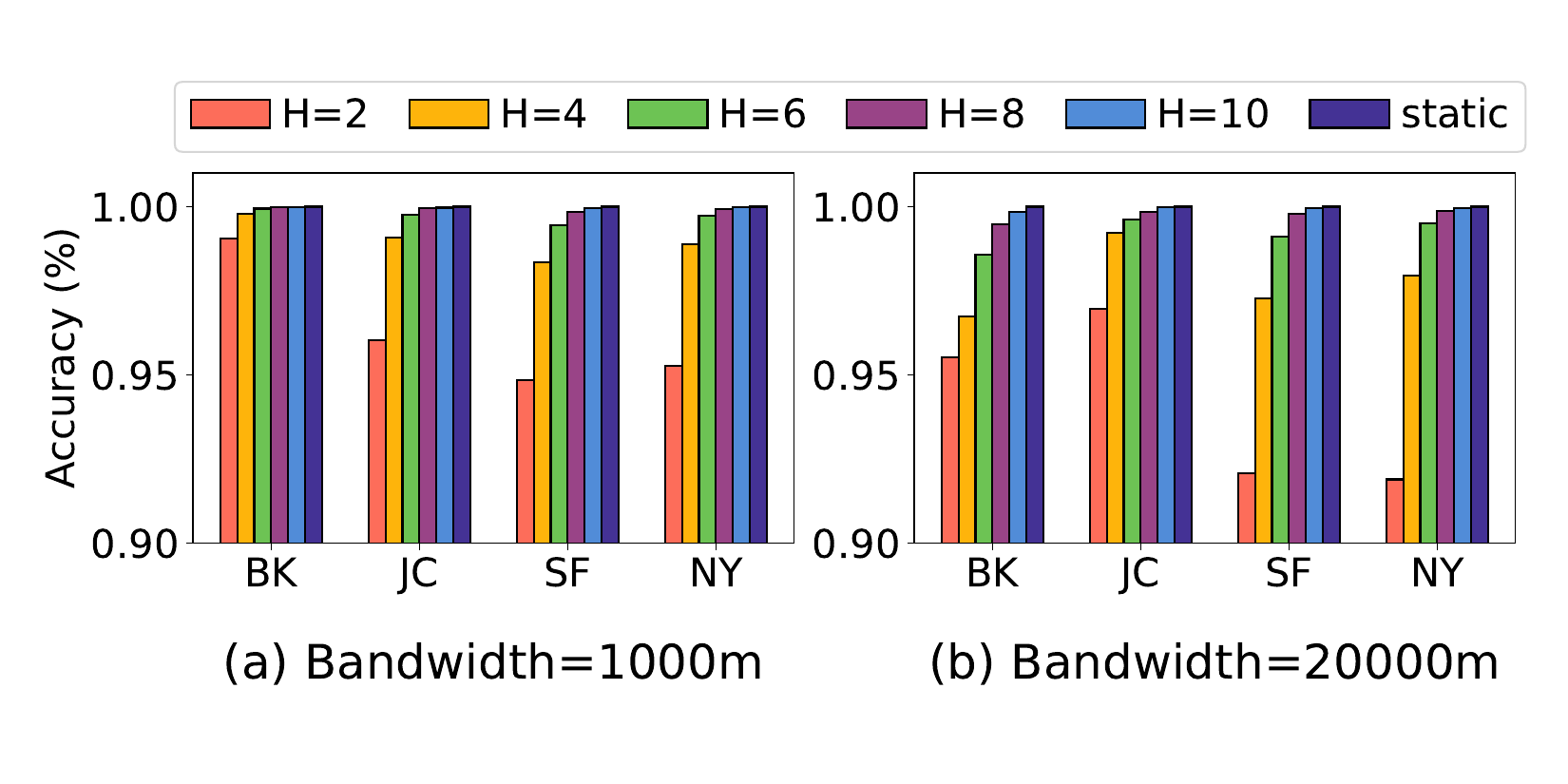}}
	\vspace{-20pt}
	\caption{Accuracy for different depth $H$.}
	\label{exp2.3}
\end{figure}

\textbf{Efficiency.} To illustrate the time cost in indexing and processing, Figure \ref{exp2.1} and Figure \ref{exp2.2} separately show the indexing time and processing time, respectively. Even when $H=10$, DRFS does not require too much time to build the index and is even faster than RFS on a large dataset(i.e., NewYork). Approximately, RFS is equivalent to DRFS with $H {\approx} 8$. In terms of the processing time, there is no significant difference with different $H$. While the indexing part is only executed once in a query batch, DRFS performs well in efficiency even with a larger $H$.

\textbf{Accuracy.} DRFS is an approximate solution, so we also need to take accuracy into consideration. All experiments reported highly accurate results shown in Figure \ref{exp2.3}. Even with $H{=}2$, the accuracy is still over $94\%$. Furthermore, accuracy increases rapidly with $H$, and when $H{=}10$, accuracies in all four datasets are greater than $99.9\%$.

\textbf{Memory.} Figure \ref{exp2.4} indicates the memory consumption, which reflects the index size. The trend is similar to Figure \ref{exp2.1} and the equivalent boundary is also $H {\approx} 8$. When $H{=}2$, the memory consumption is near ADA, indicating the feasibility of our quantization strategy. Moreover, the growth ratio is almost linear and can be predicted, which can help users to choose an appropriate parameter $H$.

Taking into account all these aspects, the experiment demonstrates that DRFS is highly practical. A small value of $H$ saves a significant amount of time and still yields a relatively accurate solution, whereas a large value of $H$ requires only slightly more time to produce nearly identical results. Therefore, users can choose the appropriate value of $H$ and adjust it dynamically.

\subsection{Kernel Function}

	Another feature of our framework is the replaceable kernel function. Since all computation processes with different kernel functions have a $O(1)$ complexity, they can generate a heatmap in the same amount of time.
	
	Figure \ref{fig:exp_kernel} illustrates three distinct heatmaps using the Triangular, Cosine, and Exponential kernel functions. Each heatmap is accompanied by its corresponding function graph positioned above. It is noteworthy that the value of the Cosine function is always greater than that of the Exponential function and then the Triangular function, thus all density values are normalized. Generally, the slope order from high to low is Triangular, Exponential, and Cosine, which is related to the smoothness of the results. Furthermore, the Triangular kernel function exhibits no truncation, whereas both the Cosine and Exponential kernel functions are truncated at $\pm 1$. This leads to a more concentrated result and distinct boundaries.

	\begin{figure}[t!]\centering
		\scalebox{0.3}[0.3]{\includegraphics{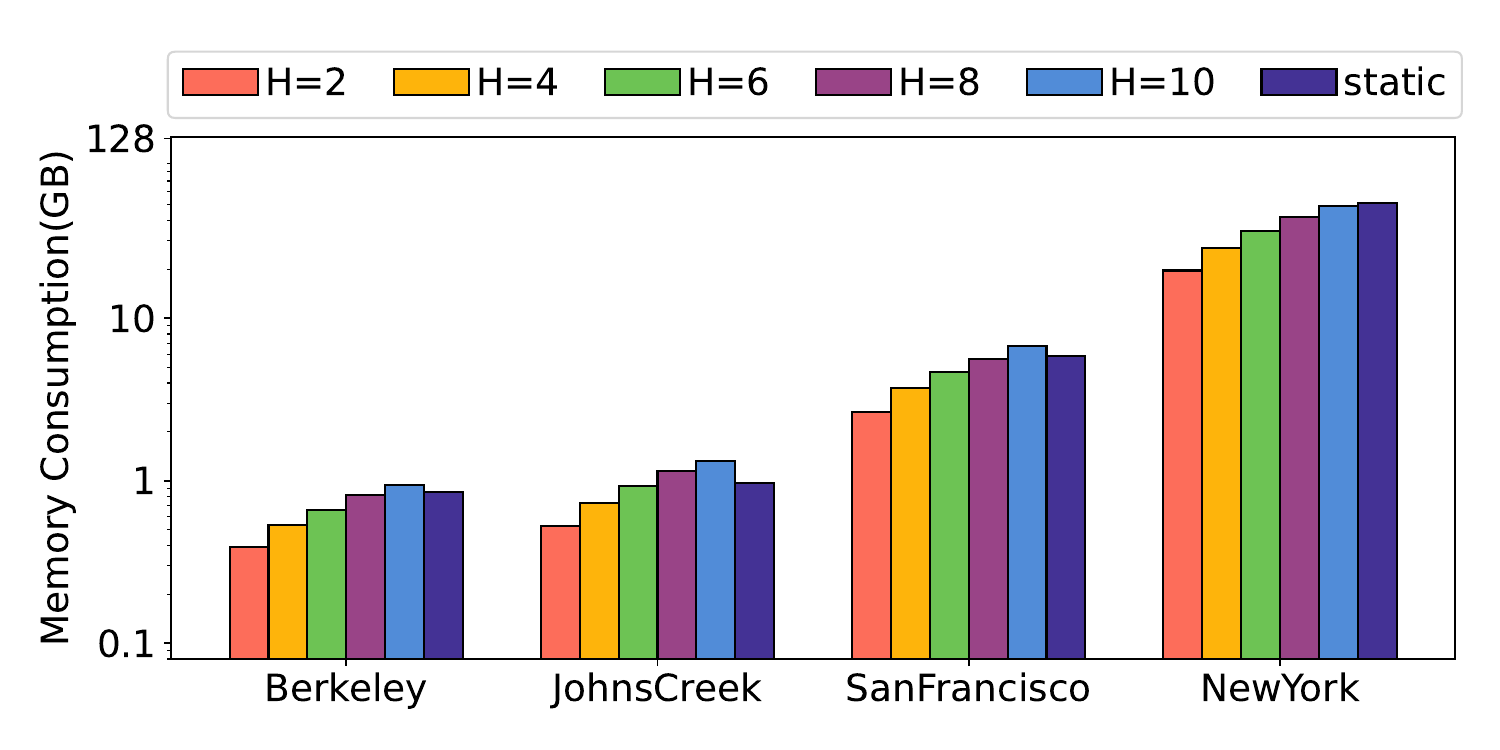}}
		\vspace{-8pt}
		\caption{Memory consumption for different depth $H$.}
		\label{exp2.4}\vspace{-1ex}
	\end{figure}

\begin{figure}[t!]\centering
	\begin{minipage}{0.3\linewidth}
	\centering\includegraphics[height=10mm, width=25mm]{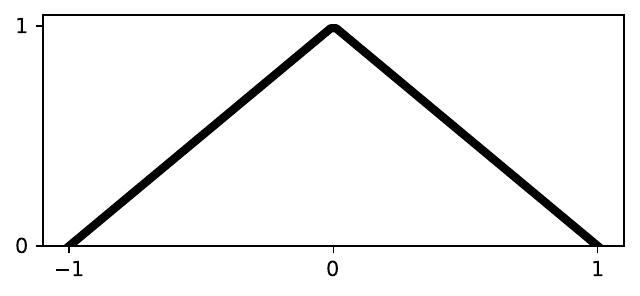}
	\end{minipage}
	\begin{minipage}{0.3\linewidth}
	\centering\includegraphics[height=10mm, width=25mm]{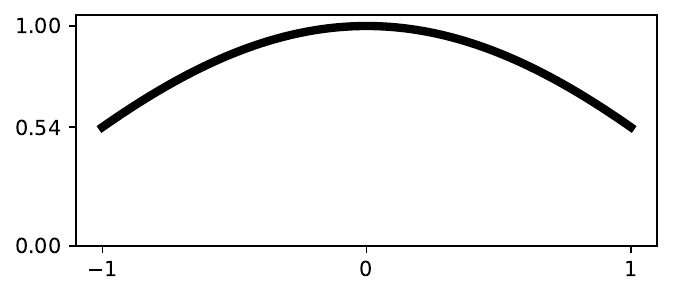}
	\end{minipage}
	\begin{minipage}{0.3\linewidth}
	\centering\includegraphics[height=10mm, width=25mm]{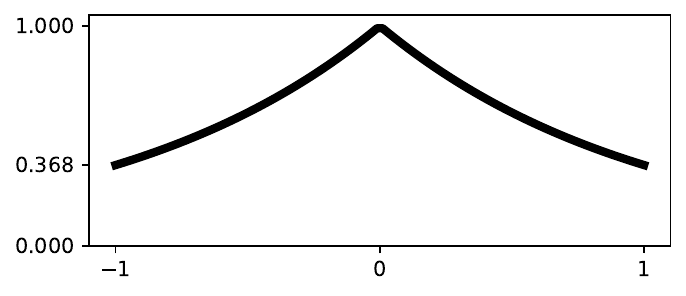}
	\end{minipage}
	
	\begin{minipage}{0.3\linewidth}
	\centering\includegraphics[height=32mm, width=24mm]{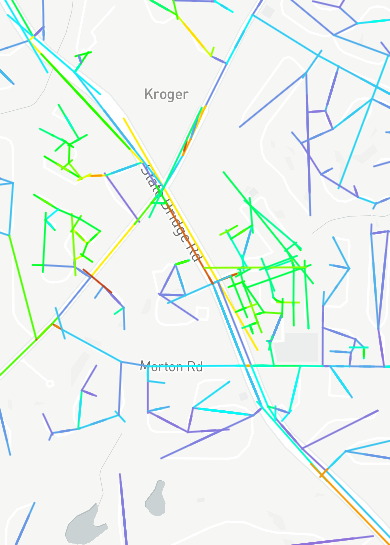}
	\caption*{(a) Triangular}
	\end{minipage}
	\begin{minipage}{0.3\linewidth}
	\centering\includegraphics[height=32mm, width=24mm]{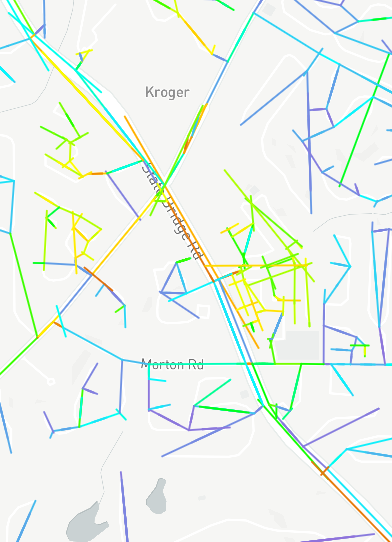}
	\caption*{(b) Cosine}
	\end{minipage}
	\begin{minipage}{0.3\linewidth}
	\centering\includegraphics[height=32mm, width=24mm]{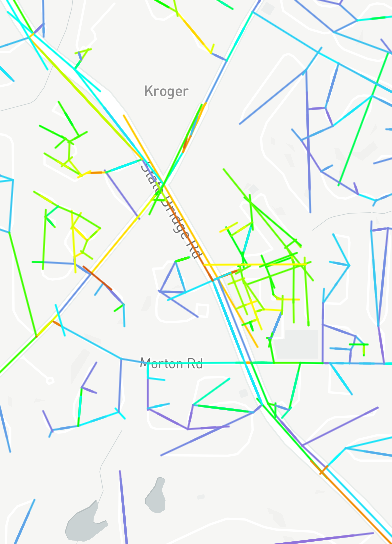}
	\caption*{(c) Exponential}
	\end{minipage}
	\vspace{8pt}
	\caption{TN-KDE result with different kernel functions.}
	\label{fig:exp_kernel}\vspace{-2ex}
\end{figure}

\noindent\textbf{Summary:} Our proposed method RFS shows significant speedup, up to 6 and 88.9 times compared to ADA and SPS, respectively, on responding to multiple online TN-KDEs. This improvement would be even more impressive with larger datasets and higher resolution requirements. RFS requires more space to store a large index, but it only incurs an 8-fold increase in memory usage and remains stable with more events. DRFS achieves a similar time and space complexity as RFS while maintaining an accuracy of over 99.9\% in most cases. When the index is quantized at $H=2$, DRFS can further save up to 40\% in time costs and 60\% in memory costs. We also tested other kernel functions and visualized their results, which matched well in high-density areas but differed in boundary regions.

%% file: conclusion.tex
\section{Conclusion}

	In this paper, we introduce the Range Forest Solution (RFS) for generating a heatmap on road networks with spatial and temporal data. 
	We also develop the Dynamic Range Forest Solution (DRFS) to support the dynamic structure and the insertion operation.
	Additionally, an optimization called Lixel Sharing is applied, which can share a similar KDE values between two adjacent lixels. 
	Finally, the kernel function in our solutions is also replaceable to generate an accurate result using either the Exponential or Cosine function.
	